\documentclass[reprint,superscriptaddress]{revtex4-2}
\usepackage{amssymb,amsmath,amsthm}
\usepackage{graphicx}
\usepackage{xcolor}
\usepackage{bm,bbm}
\usepackage[bookmarks=false]{hyperref}
\hypersetup{colorlinks=true,citecolor=blue,linkcolor=red,%
urlcolor=blue,pdfstartview=FitH,bookmarksopen=true}
\usepackage{comment}
\usepackage{mathtools}
\usepackage{enumitem}
\usepackage{algorithm}
\usepackage{algorithmicx}
\usepackage{algpseudocode}
\usepackage{subfigure}%
\usepackage{parskip}
\usepackage{contour}
\usepackage{geometry}
\newcommand{\ket}[1]{\text{$ | #1 \rangle $}}
\newcommand{\bra}[1]{\text{$ \langle #1 | $}}

\newcommand{\tr}{\mathrm{Tr}}

\contourlength{0.1pt}

\newtheoremstyle{note}
  {\topsep/2}              	% ABOVE SPACE
  {\topsep/2}            	% BELOW SPACE
  {}                        % BODY FONT
  {\parindent}             	% INDENT (empty value is the same as 0pt)
  {\itshape}                % HEAD FONT
  {.---}                    % HEAD PUNCTUATION
  {0pt}                     % HEAD SPACE
  {\thmname{#1}\thmnumber{ \itshape#2}\thmnote{ (#3)}} % CUSTOM-HEAD-SPEC

\newtheorem{theorem}{Theorem}
\newtheorem{lemma}{Lemma}

\theoremstyle{definition}

\theoremstyle{remark}

\makeatletter
\newsavebox{\@brx}
\newcommand{\llangle}[1][]{\savebox{\@brx}{\(\m@th{#1\langle}\)}%
  \mathopen{\copy\@brx\kern-0.5\wd\@brx\usebox{\@brx}}}
\newcommand{\rrangle}[1][]{\savebox{\@brx}{\(\m@th{#1\rangle}\)}%
  \mathclose{\copy\@brx\kern-0.5\wd\@brx\usebox{\@brx}}}
\makeatother
\begin{document}
\title{Estimating many properties of a quantum state via quantum reservoir processing}

\author{Yinfei Li}
\affiliation{Key Laboratory of Advanced Optoelectronic Quantum Architecture and Measurement of Ministry of Education, School of Physics, Beijing Institute of Technology, Beijing 100081, China}

\author{Sanjib Ghosh}
\email{sanjibghosh@baqis.ac.cn}
\affiliation{Beijing Academy of Quantum Information Sciences, Beijing 100193, China}

\author{Jiangwei Shang}
\email{jiangwei.shang@bit.edu.cn}
\affiliation{Key Laboratory of Advanced Optoelectronic Quantum Architecture and Measurement of Ministry of Education, School of Physics, Beijing Institute of Technology, Beijing 100081, China}

\author{Qihua Xiong}
\affiliation{Beijing Academy of Quantum Information Sciences, Beijing 100193, China}
\affiliation{State Key Laboratory of Low-Dimensional Quantum Physics and Department of Physics, Tsinghua University, Beijing 100084, China}

\author{Xiangdong Zhang}
\affiliation{Key Laboratory of Advanced Optoelectronic Quantum Architecture and Measurement of Ministry of Education, School of Physics, Beijing Institute of Technology, Beijing 100081, China}

% %
\date{\today}
% %

%%%%%%
\begin{abstract}
Estimating properties of a quantum state is an indispensable task in various applications of quantum information processing. To predict properties in the post-processing stage, it is inherent to first perceive the quantum state with a measurement protocol and store the information acquired.
In this work, we propose a general framework for constructing classical approximations of arbitrary quantum states with quantum reservoirs.
A key advantage of our method is that only a single local measurement setting is required for estimating arbitrary properties, while most of the previous methods need exponentially increasing number of measurement settings.
To estimate $M$ properties simultaneously, the size of the classical approximation scales as $\ln M$.
Moreover, this estimation scheme is extendable to higher-dimensional systems and hybrid systems with non-identical local dimensions, which makes it exceptionally generic.
We support our theoretical findings with extensive numerical simulations. 
\end{abstract}

\maketitle

%%%%%%
\section{Introduction}
Estimating properties of a quantum state plays a central role in the implementation of various quantum technologies, such as quantum computing, quantum communication and quantum sensing.
This highlights that extracting information from a quantum system to a classical machine lies at the heart of quantum physics~\cite{Gebhart2023}. The prominent technique for this task, quantum tomography, studies the reconstruction methods of density matrix, which captures all the information of a quantum system. However, the curse of dimensionality has emerged with the advent of the noisy intermediate-scale quantum (NISQ) era~\cite{Preskill2018}, which renders it infeasible to obtain a complete description of quantum systems with a large number of constituents. Moreover, a full description is often superfluous in tasks where only key properties are relevant. As a consequence, the concept of shadow tomography is proposed to focus on predicting certain properties of a quantum system~\cite{Aaronson10.1145/3188745.3188802}.

A particularly important progress in the study of shadow tomography is the advancement of randomized measurements~\cite{Marco2019,Huang2020,PhysRevA.99.052323}, the virtue of which is highlighted as
\textit{``Measure first, ask questions later''} \cite{RandMeasTb}. The randomized measurement protocols proposed by Huang, Kueng and Preskill construct approximate representations of the quantum system, namely classical shadows, via Pauli group and Clifford group measurements~\cite{Huang2020}. The single-snapshot variance upper bound of classical shadows is determined by the so-called shadow norm, which is optimal in the worst-case scenario. Nevertheless, variants of randomized measurements can achieve a better performance in specific cases~\cite{RSPRXQuantum.2.030348,PhysRevLett.127.200501,Ham_PhysRevResearch.4.013054}. For instance, the Hamiltonian-driven shadow tomography achieves a higher efficiency in predicting diagonal observables~\cite{Ham_PhysRevResearch.4.013054}.
In addition, the readout noise present in the measurement protocol can be further suppressed by constructing classical shadows with positive operator-valued measures (POVMs)~\cite{SdPOVM_PhysRevLett.129.220502,POVM2_PhysRevLett.130.100801,POVM3_PRXQuantum.2.040342}.

The classical shadows are highly efficient in the estimation of various properties in the post-processing phase, the benefits of which extend to entanglement detection~\cite{Entdect_PhysRevLett.125.200501}, characterization of topological order~\cite{TPorder_doi:10.1126/sciadv.aaz3666}, machine learning for many-body problems~\cite{MLfMB_Huang2022}, etc. 
However, the randomized measurements protocols pose a challenge in experiments due to the need for exponentially increasing measurement settings to achieve an arbitrary accuracy. Hence, various techniques are introduced to tackle this problem~\cite{Derand_PhysRevLett.127.030503,Single_SettingPRXQuantum.3.040310,PhysRevX.13.011049,mcginley2022shadow}.
Moreover, the theoretical results are based on the fact that multi-qubit Clifford groups are unitary 3-designs~\cite{DesignPhysRevA.96.062336}, which is not the case for arbitrary qudit systems. The generalization of these results to higher-dimensional systems typically require complex unitary ensembles that are relatively hard to implement~\cite{QuditShadow,SOSPS_grier2022sampleoptimal}. 
Therefore, a general method for direct estimation with a single measurement setting is highly desirable.

Recently, quantum neural networks~\cite{schuld2014,biamonte2017,schuld2021} are widely studied as promising artificial neural networks due to their enhanced information feature space supported by the exponentially large Hilbert space~\cite{havlivcek2019,abbas2021}. Unlike traditional computing frameworks, neural networks learn to perform complex tasks based on training rather than predefined algorithms or strategies~\cite{jain1996}. 
With the capacity to produce data that displays atypical statistical patterns, quantum neural networks have the potential to outperform their classical counterparts~\cite{biamonte2017}.
However, training a quantum neural network can be equally hard~\cite{cerezo2022}. Indeed, it has been shown that training of quantum neural networks could be exceptionally difficult owing to the barren plateaus or far local minima in the training landscapes~\cite{McClean2018,cerezo2021,bittel2021,wang2021}. This is the reason that quantum neural networks are often limited to shallow circuit depths or small number of qubits. A trending line of research that circumvents this issue is quantum reservoir processing~(QRP) or reservoir computing~\cite{Ghosh2019,QRSP_PhysRevLett.123.260404,review_Nakajima2020,PhysRevX.11.041062,ReservoirComputing}, which studies the quantum analogy of recurrent networks. Note that in this context, `reservoir' refers to a type of neural network. In QRP, training is completely moved out of the main network to a single output layer, such that the training becomes a linear regression eliminating the possibility of producing barren plateaus or local minima~\cite{Ghosh2019}. Such a quantum neural network retains its quantum enhanced feature space while being trainable via a fast and easy mechanism.

In this work we present a shadow estimation scheme on the QRP platform, which  overcomes the obstacles faced by randomized measurement protocols by harnessing the richness of QRP. While QRP serves as a quantum analog of reservoir computing, our focus in this study is limited to its role as a quantum state processing platform that extracts properties of the input quantum system. Therefore, the training at the output layer maintains as a straightforward approach by utilizing linear inversion.
A scheme of minimal quantum hardware comprising pair-wise connected quantum nodes is developed to estimate many properties of a quantum state. We highlight that only two-node training data is required, which captures the pair-wise interacting reservoir dynamics. As major advantages, our scheme requires single-node measurements, only in a single setting, and a linear reservoir size $2n$ with respect to the number of constituents $n$ of the input state.
All of these are particularly favorable for actual physical implementations. 

Furthermore, we establish rigorous performance guarantee by adopting the mindset of shadow estimation. According to Born's rule, one measurement of a quantum state is analogous to sampling a probability distribution once. Thus, learning properties of a quantum state involves measuring identical and independently distributed (i.i.d.) samples of the quantum state a certain number of times. To estimate $M$ observables of the state within an additive error $\epsilon$ and with constant confidence, the number of i.i.d. input samples consumed scales as $O\bigl({F}_\text{res} \ln M  /\epsilon^2\bigr)$. The factor ${F}_\text{res}$ represents the variance upper bound of the single sample estimator, which depends solely on the observables and the reservoir dynamics, and its magnitude is comparable to that of the shadow norm. As a direct consequence of the pair-wise reservoir dynamics, ${F}_\text{res}$ for a $k$-local tensor product observable is the product of that for each single-qubit observable. 
We support the theoretical results with extensive numerical simulations.

%%%%%%
\section{Quantum reservoir property estimation}\label{Sec.II}
In this section we introduce the scheme of quantum reservoir property estimation~(QRPE), which evaluates  physical properties of input quantum states based on a quantum reservoir processing device. 
Our goal parallels that of shadow estimation: to devise a resource-efficient classical representation of complex quantum states that permits access to their properties through subsequent classical processing.

\begin{figure}
    \includegraphics[width=0.98\columnwidth]{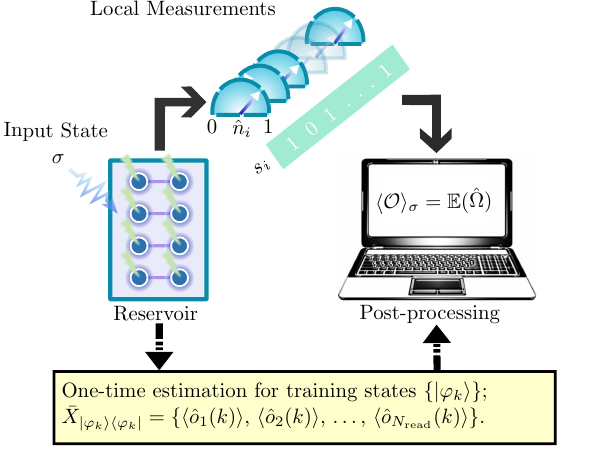}
\caption{Schematic illustration of predicting many properties with quantum reservoir processing.  A source prepares an $n$-qubit quantum state $\sigma$ which is taken as input by a pair-wise reservoir network with $2n$ nodes. For the $i$-th input copy, the local measurement operator on the $j$-th node is ${\hat{n}_j = (\openone-\varsigma_{j}^z)/{2}}$, and the readout is a bit string  ${s_i\in{\{0,1\}}^{n}}$, corresponding to one of the ${N_\text{read}}$ readout operators. Equipped with the training data, an unbiased estimator $\hat{\Omega}$ is constructed for the property $\mathcal{O}$.}
\label{fig: scheme}
\end{figure}

%%%%%%
\subsection{Physical setting}\label{sec:QRPE}
As a proof-of-principle, we first consider a dynamical estimation device based on pair-wise connected interacting qubits, as shown in Fig.~\ref{fig: scheme}, which is also known as quantum registers~\cite{Regis_PhysRevLett.93.150501,Regis_doi:10.1126/science.1139831}. Later in Sec.~\ref{sec: CVres} we will demonstrate that our scheme works also for continuous variable reservoir systems evolving under open quantum dynamics.  The connections between quantum nodes are obtained with transverse exchange interactions~\cite{majer2007coupling} and each qubit is excited with an onsite driving field.  For $n$-qubit input states, there are $n$ pairs of reservoir nodes. The corresponding Hamiltonian of the device is given by
\begin{equation}\label{eq: Hamiltonian}
    \begin{aligned}
        \hat{H} =  &\sum_{  i =1}^n \Bigl[ J \bigl( \varsigma_{2i-1}^x \varsigma^x_{2i} + \varsigma_{2i-1}^y \varsigma^y_{2i} \bigr) + P_1 \varsigma_{2i-1}^x \Bigr. \\  & +  E_1 \varsigma_{2i-1}^z
         \Bigl. + P_2 \varsigma_{2i}^x + E_2 \varsigma_{2i}^z \Bigr]\,.
    \end{aligned}
\end{equation}
The operators $\varsigma_i^{x,y,z}$ represent the Pauli operators on the $i$-th quantum node, which has a compatible dimension with the context. The parameter $J$  represents the strength of the pair-wise transverse exchange interaction between the reservoir nodes. The parameters $P_{1,2}$ and $E_{1,2}$ represent driving field strength and onsite energy respectively. 
Such a device can be readily realized with superconducting qubits, where the exchange interaction can be realized via a cavity quantum bus~\cite{majer2007coupling}. Moreover, the form of $\hat{H}$ is of a quantum spin Hamiltonian which can be realized in a variety of platforms such as NMR~\cite{alvarez2015localization,kusumoto2021experimental}, quantum dots~\cite{kandel2021adiabatic}, and trapped ions~\cite{porras2004effective}.

\subsection{Measurement protocol}\label{sec:QRPE-MP}
The quantum state of interest, which we denote by $\sigma$, is injected into the quantum reservoir via an invertible map. For simplicity, we consider the {\sc swap} operations. For a pair of interacting nodes ${\langle 2i-1, 2i\rangle}$, the $(2i-1)$-th reservoir node is connected to the input system.  Thus, the initial state of the network at time ${t=0}$ is given by the density matrix:
${\rho(0) = \sigma \otimes [|0\rangle \langle 0 |]_\text{rest}}$,
where the suffix ‘rest’ indicates the network nodes other than the ones connected to the input qubits via the {\sc swap} gates. 
The initial reservoir state evolves in time as,
\begin{eqnarray}\label{eq:reservoirstate}
\rho(t) = \hat{U}^\dagger (t) \rho(0) \hat{U} (t)\,,
\end{eqnarray}
where $\rho(t)$ is the density operator at time $t$ and $\hat{U}(t)=\exp(-it\hat{H}/\hbar)$ is the evolution operator. After a sufficient time evolution, we perform local Pauli-$Z$ measurements on the reservoir nodes (qubits). For each node there are two readouts, $+1$ and $-1$, represented by the projectors onto the positive and negative eigen-subspace of Pauli-$Z$ operator $\{\varsigma_{i}^z\}$ respectively. The final readouts are provided by a set of commuting readout operators
\begin{equation}\label{eq:readout_operators}
    \begin{aligned}
\{\hat{o}_i\}& =\prod_{j = 1}^{N_\text{node}}\Bigl\{\frac{\openone-\varsigma_{j}^z}{2}, \,\frac{\openone+\varsigma_{j}^z}{2}\Bigr\} \\
& = \bigl\{\hat{C}_{\varnothing}\bigr\}\cup \bigl\{\hat{C}_{\{i_1\}}\bigr\}\cup\bigl\{\hat{C}_{\{i_1',i_2'\}}\bigr\}\cup\dots\,,
    \end{aligned}
\end{equation} 
where $N_\text{node}$ is the number of reservoir nodes. $\hat{C}_S$ is an element of the set $\{\hat{o}_i\}$, which is related to the Pauli-$Z$ operators as 
\begin{eqnarray}
\hat{C}_{S} = \prod_{i\in S}\frac{\openone-\varsigma_{i}^z}{2}\prod_{j\notin S}\frac{\openone+\varsigma_{j}^z}{2}\,,
\end{eqnarray}
so it represents the configuration of measurement outcomes that only nodes in the set ${ S}$ result in $-1$ for local Pauli-$Z$ measurements. Hence, the readout for the $i$-th input is recorded by a bit string $s_i\in\{0,1\}^n$.
The total number of readout operators in $\{\hat{o}_i\}$ is given by
\begin{eqnarray}
N_\text{read}= \sum_{k=0}^{N_\text{node}}\tbinom{N_\text{node}}{k} = 2^{N_\text{node}}\,,
\end{eqnarray} 
where $\tbinom{N_\text{node}}{k}$ is the combinatorial number. 

Considering that $\{\hat{o}_i\}$ forms an orthogonal measurement basis with a total number of $2^{N_\text{node}}$ elements, a reservoir with a minimal of ${N_\text{node} = 2\log_2 d}$ nodes is required to estimate arbitrary quantum properties for input states supported on a $d$-dimensional Hilbert space $\mathcal{H}_d$. This observation agrees with our proposal of pair-wise connected reservoir networks. Moreover, the commutativity of the chosen readout operators leads to quantum resource effectiveness. This is in sharp contrast to the traditional quantum reservoir computing schemes where either the required size of the quantum reservoir or the temporal resolution in the measurement tend to be exponentially large ($\sim\!d^2$). In either situations, these traditional schemes are exponentially quantum resource consuming.

%%%%%%
\subsection{Training}\label{sec:train}

Here we introduce a training process that captures the internal maps of QRP.
In the context of QRP, the essential condition is that we have reproducible reservoir dynamics. However, the reservoir could be largely a black box, especially when we consider open quantum dynamics. Thus, a training process at a single output layer is required.
The basic property of a quantum state $\sigma$ is the linear function in the form of $\tr{(\mathcal{O}\sigma)}$, where $\mathcal{O}$ is a linear operator that is compatible with $\sigma$. For instance, the expectation value of an arbitrary observable $\mathcal{O}$ takes on this form.
To estimate the linear function $\tr{(\mathcal{O}\sigma)}$, the QRPE scheme essentially maps the input state $\sigma$ to a vector of probabilities for observing each readout operator
\begin{equation}\label{eq:statemap}
    \sigma\xrightarrow{\text{QRP}} \bar{X} = [\langle\hat{o}_1\rangle;\,\langle\hat{o}_2\rangle;\,\dots;\,\langle\hat{o}_{N_\text{read}}\rangle]\,,
\end{equation}
and the target observable $\mathcal{O}$ to a vector of weights
\begin{equation}\label{eq:obsmap}
    \mathcal{O}\xrightarrow{\text{QRP}}W = [w_1,\,w_2,\,\dots,\,w_{N_\text{read}}]\,,
\end{equation}
satisfying
\begin{equation}\label{eq:unbiasedest}
    W\cdot \bar{X} = \tr(\mathcal{O}\sigma)\,.
\end{equation}
We note that similar maps have also been studied in the context of analog quantum simulation~\cite{PhysRevX.13.011049,mcginley2022shadow}.
Here each readout in experiments requires only linear classical storage with respect to the system size, owing to the tensor product structure in Eq.~\eqref{eq:readout_operators}. The relation given by Eq.~\eqref{eq:statemap} is achieved by sampling from i.i.d. copies of the  $n$-qubit state $\sigma$ and processing the reservoir readouts with statistical methods, as addressed in Sec.~\ref{sec:ResEst}, while Eq.~\eqref{eq:obsmap} is achieved by a training process described below. 

For training, we require a one-time estimation of a known set of training states $\{\ket{\varphi_k}\}$.
Here we consider the training data to be accurate, and present results that account for statistical noise occurring outside of the training phase.
The reservoir dynamics are initialized by setting the parameters $J$, $P_{1,2}$, $E_{1,2}$ and the evolution time $t$. 
Each training step starts with inputting a training state $|\varphi_k\rangle\langle\varphi_k | $ into the reservoir to reach the initial training state
${\rho(0) = |\varphi_k\rangle\langle\varphi_k | \otimes[ |0\rangle\langle 0|]_\text{rest}}$,
which then evolves to $\rho(t)$ at time $t$.
From sufficiently many measurement results of each input training state, we estimate the expectation value $\langle\hat{o}_i (k)\rangle$ of the readout operator.
The training data is stored by arranging $\langle\hat{o}_i (k)\rangle$ into a column vector $\bar{X}_{\ket{\varphi_k}\bra{\varphi_k}}$.
In this way, we collect readout vectors $\bar{X}_{\ket{\varphi_k}\bra{\varphi_k}}$ corresponding to all the training states
$|\varphi_k\rangle$ for $k=1,\,2,\,\dots,\,N_\text{train}$. For the simplest qubit-node reservoir system as described in this section, there is no setback in having a complete characterization over the reservoir dynamics. However, when the reservoir dynamics is much more complex as discussed in Sec.~\ref{sec: CVres}, the training process described here can still capture the reservoir maps with a resource cost depending only on the dimension of input quantum system. We note that the training process is similar to a quantum process tomography over a small quantum system.

Until this step, the whole procedure is completely independent of the property to be estimated.
The knowledge of $\mathcal{O}$ is only required at the post processing level,
where we set the target output 
\begin{equation}\label{eq:tragetvector}
    Y^\text{tar}_k = \langle\varphi_k|\mathcal{O}|\varphi_k\rangle\,.
\end{equation}
Let ${Y^\text{out}_k = W\cdot\bar{X}_{\ket{\varphi_k}\bra{\varphi_k}}}$, and the sum of squared deviations between $Y^\text{out}_k$ and $Y^\text{tar}_k$ is
\begin{eqnarray}
\mathcal{E}_\text{train} = \bigl|Y^\text{out}-Y^\text{tar}\bigr|^2\,,
\end{eqnarray}
where $Y^\text{out}$ is a row vector with elements $Y^\text{out}_k$, ${k = 1,\,2,\,\dots,\,N_\text{train}}$, and $Y^\text{tar}$ is defined similarly.
For typical quantum reservoirs, a total of $d^2$ training states are needed to estimate arbitrary properties for an input state supported on $\mathcal{H}_d$. The set of training states $\{{\ket{\varphi_k}\bra{\varphi_k}}\}$ is a set of $d^2$ vectors which spans a $d^2$-dimensional Hilbert space, i.e., forms an informationally complete POVM. We choose the training states as
\begin{eqnarray}\label{eq:trainingstate}
 |\varphi_k \rangle =\otimes_{m=1}^n |k_m\rangle\,, \quad  k=\sum_{m=1}^n4^{m-1}k_m\,.
\end{eqnarray}
where ${k_m\in \{0,1,2,3\}}$, ${|0\rangle= [1;0]}$, ${|1\rangle = [0; 1]}$, $|2\rangle ={(|0\rangle + |1\rangle )/\sqrt{2}}$ and ${|3\rangle =(|0\rangle + i|1\rangle)/\sqrt{2}}$.
Denote 
\begin{equation}
    \bm{\mathcal{T}} = {\bm{X}_\text{t}}\bm{M}_\text{t}^{-1}\,,
\end{equation}
where ${\bm{M}_\text{t} = \bigl[|\varrho_1\rrangle,\,|\varrho_2\rrangle,\,\dots,\,|\varrho_{d^2}\rrangle\bigr]}$, $\varrho_k = \ket{\varphi_k}\bra{\varphi_k}$ is the density matrix of the $k$-th training state, $|\cdot\rrangle$ is the Louiville superoperator representation, and the matrix of training data is $\bm{X}_\text{t} = [\bar{X}_{\ket{\varphi_1}\bra{\varphi_1}},\,\bar{X}_{\ket{\varphi_2}\bra{\varphi_2}},\,\dots]$. Then the expectation of reservoir readout for an input state $\sigma$ is
\begin{equation}\label{eq:14}
\begin{aligned}
    \bar{X}_\sigma &= \bm{\mathcal{T}}|\sigma\rrangle\,,
\end{aligned}
\end{equation}
as is explained in Appendix.~\ref{app:geointerp}.
If $\bm{X}_\text{t}$ is full rank, there exists a weight vector $W$ that minimizes $\mathcal{E}_\text{train}$, i.e.,
\begin{equation}
W = Y^\text{tar}{\bm{X}_\text{t}}^{-1} \,,
\end{equation}
which can be written as $W = \llangle\mathcal{O}|\bm{\mathcal{T}}^{-1}$. 

Leveraging pair-wise reservoir dynamics in Eq.~\eqref{eq: Hamiltonian} lifts the burden on training, given that the training data inherits a tensor product structure. Once the training data of a single pair of interacting reservoir nodes is collected as $\bm{X}_\text{p}$, then we have
\begin{equation}
    \bm{X}_\text{t} = \bigotimes_{i = 1}^n \bm{X}_\text{p}\,,
\end{equation}
where $n$ is the number of node pairs.
Also, 
\begin{equation}
    \bm{M}_\text{t} = \bigotimes_{i = 1}^n \bm{M}_\text{p}\,,\quad\bm{\mathcal{T}} =  \bigotimes_{i = 1}^n \bm{\mathcal{T}}_\text{p}\,,
\end{equation}
where $\bm{\mathcal{T}}_\text{p} = {\bm{X}_\text{p}}\bm{M}_\text{p}^{-1}$.
Thus, the training task is effectively reduced to that of a two-node reservoir, and the full-rank requirement of $\bm{X}_\text{t}$ is correspondingly reduced to that of $\bm{X}_\text{p}$.  Moreover, the vector of weights only necessitates polynomial storage if the property $\mathcal{O}$ can be decomposed into a finite sum of tensor products. See Appendix.~\ref{app:geointerp} for more details.

%%%%%%
\subsection{Reservoir estimator}\label{sec:ResEst}
In this section we introduce the reservoir estimators for linear functions. With the training data at our disposal, we could analyze the sample efficiency of the reservoir estimators.

Suppose the observed readout operator for the $i$-th input copy is $\hat{o}_j$, then the so-called single snapshot $X_i$ is a vector where the $j$-th element is $1$ and the other elements are $0$. 
For a total of $N_\text{sample}$ input copies, one obtains a set of snapshots $\{X_i\,|\,i=1,\,2,\,\dots,\,N_\text{sample}\}$.
The single-snapshot estimator is
\begin{equation}
    \hat{\Omega} \equiv  W\cdot \hat{X}\,,
\end{equation}
where $\hat{X}$ is a random variable that conforms to the probability distribution behind $X_i$. Eq.~\eqref{eq:unbiasedest} indicates that $\hat{\Omega}$ is an unbiased estimator for $\tr(\mathcal{O}\sigma)$.
Hence in data processing, we could apply the median of means (MoM) method to neutralize the effect of outliers~\cite{Huang2020,RSPRXQuantum.2.030348}.
After processing ${N_\text{sample} = KN}$ input copies, we divide the snapshots into $K$ equally sized subsets $\{X_i^v\,|\,i=1,\,2,\,\dots,\,N\}$, and compute the mean value of the single-snapshot estimators for each subset. The corresponding estimators are
\begin{equation}
    \hat{\Omega}^v_\text{M} = \frac{1}{N}\sum_{i = 1}^{N}{ W\cdot {\hat{X}_i^v}},\quad v = {1,\,2,\,\dots,\,K}\,.
\end{equation} 
Then, the MoM estimator is given by 
\begin{equation}
     \hat{\Omega}_\text{MoM} = \text{Median}\bigl\{\hat{\Omega}^v_\text{M}\bigr\}\,.
\end{equation}
With this, we have
\begin{theorem}\label{theorem1}
The number of quantum state inputs needed for estimating a set of $M$ properties ${\{\mathcal{O}_i\,|\,i = 1,\,2,\,\dots,\,M\}}$ to precision $\epsilon$ and confidence level $1-\delta$ scales in 
\begin{equation}
    N_\text{sample}\sim O\biggl(\ln\Bigl(\frac{2M}{\delta}\Bigr)\frac{\max_i {F}_\text{res}^i}{\epsilon^2}\biggr)\,.
\end{equation}
The factor ${F}_\text{res}^i = || \mathcal{B}_i||_\infty $ is the variance upper bound of the single-snapshot estimator of $\mathcal{O}_i$ maximized over the possible quantum state inputs, where $||\cdot||_\infty$ represents the spectral norm, and $\mathcal{B}_i$ is defined by
\begin{equation}
\llangle\mathcal{B}_i|{\bm{\mathcal{T}}}^{-1} = \llangle\mathcal{O}_i|{\bm{\mathcal{T}}}^{-1} \odot \llangle\mathcal{O}_i|{\bm{\mathcal{T}}}^{-1}\,.
\end{equation}
\end{theorem}
\begin{proof}
    This efficiency scaling results from the median of means method, and we consider the worst-case scenario by maximizing ${F}_\text{res}^i$ over the set of observables. A more detailed proof is included in Appendix~\ref{supp:B}.
\end{proof}

The variance of a single-snapshot estimator for $\mathcal{O}$ is invariant for properties ${\{\mathcal{O}'\,|\,\mathcal{O}'= \mathcal{O} + c\openone, c\in \mathbb{C}\}}$. Thus, one could use only the traceless part of $\mathcal{O}$ to compute the worst-case variance upper bound. It is interesting to note that the MoM estimator won't have a visible significance in some tested cases~\cite{PhysRevA.104.052418, ExpPRXQuantum.2.010307}, where it could be replaced with the sample mean estimator. The performance factor ${F}_\text{res}$ allows for an investigation of the reservoir parameters in Eq.~\eqref{eq: Hamiltonian}. To achieve tomographic completeness, it is essential to have sufficiently large evolution time $t$ and hopping strength $J$, ensuring effective information scrambling. We also find that when both $E_1$ and $E_2$ equal 0 or both $P_1$ and $P_2$ equal 0, the reservoir map will be tomographically incomplete, which is closely related to the measurement setting. If we replace $\varsigma_z$ with $\varsigma_x$ or $\varsigma_y$ in Eq.~\eqref{eq:readout_operators}, then the reservoir map can be tomographically complete when both $P_1$ and $P_2$ equal 0. However, the reservoir settings with a good performance seem to happen without a pattern in numerical experiments. The performance of one reservoir setting at different time and random reservoir settings is illustrated in Fig.~\ref{fig:Var_T} and Fig.~\ref{fig:TMP} respectively in Appendix.~\ref{supp:B}. We also find that increasing more nodes or non-local interactions is unlikely to improve the performance. On the contrary, in numerical experiments of reservoirs with random nearest-neighbor interaction or fully-connected interaction, the average performance is worse than that of the pair-wise interacting reservoir by several orders of magnitude.
 
To further analyze the sample efficiency scaling of the QRPE scheme, we have:
\begin{theorem}\label{thm:Scaling}
    For a $k$-local tensor product observable, e.g.
    \begin{equation}
    \mathcal{O} = \bigotimes_{i = 1}^k \mathcal{O}_i \bigotimes_{i = k+1}^n \openone_i\,,
\end{equation}
the worst-case variance upper bound $|| \mathcal{B}||_\infty$ is the product of that of each local observables $\mathcal{O}_i$, i.e.,
\begin{equation}
    || \mathcal{B}||_\infty
    = \prod_{i = 1}^k || \mathcal{B}_i||_\infty\,,
\end{equation}
where $\mathcal{B}_i$ satisfies
\begin{equation}
    \llangle\mathcal{B}_i|{\bm{\mathcal{T}}_\text{p}}^{-1} = \llangle\mathcal{O}_i|{\bm{\mathcal{T}}_\text{p}}^{-1} \odot \llangle\mathcal{O}_i|{\bm{\mathcal{T}}_\text{p}}^{-1}\,.
\end{equation}
\end{theorem}
\begin{proof}
    This theorem results from the pair-wise reservoir dynamics. See Appendix~\ref{supp:B} for the details.
\end{proof}

\begin{figure}[t]
	\centering
	\setcounter {subfigure} {0} {
		\includegraphics[width=0.45\textwidth]{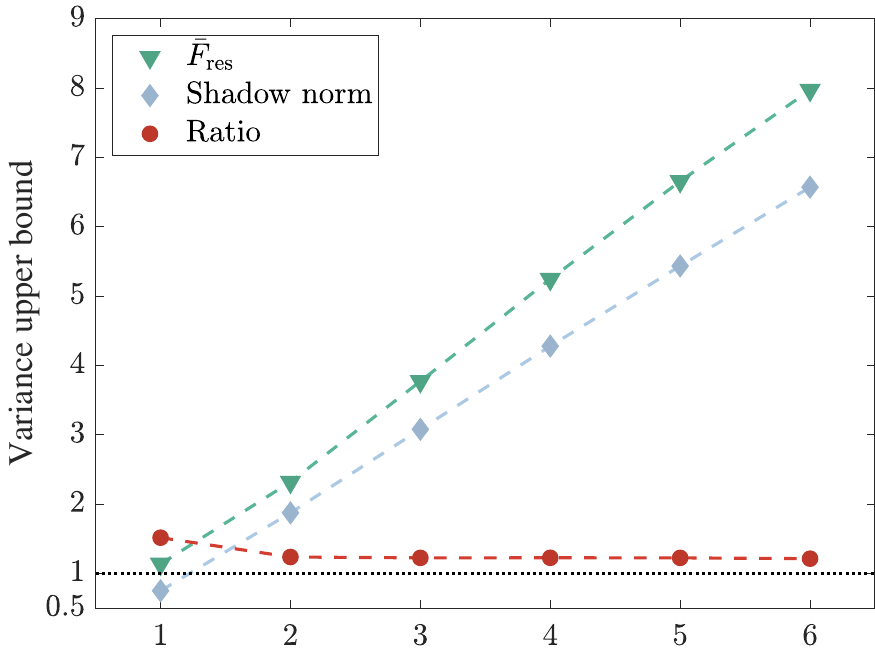}}
	\\
	\setcounter {subfigure} {0} {
		\includegraphics[width=0.45\textwidth]{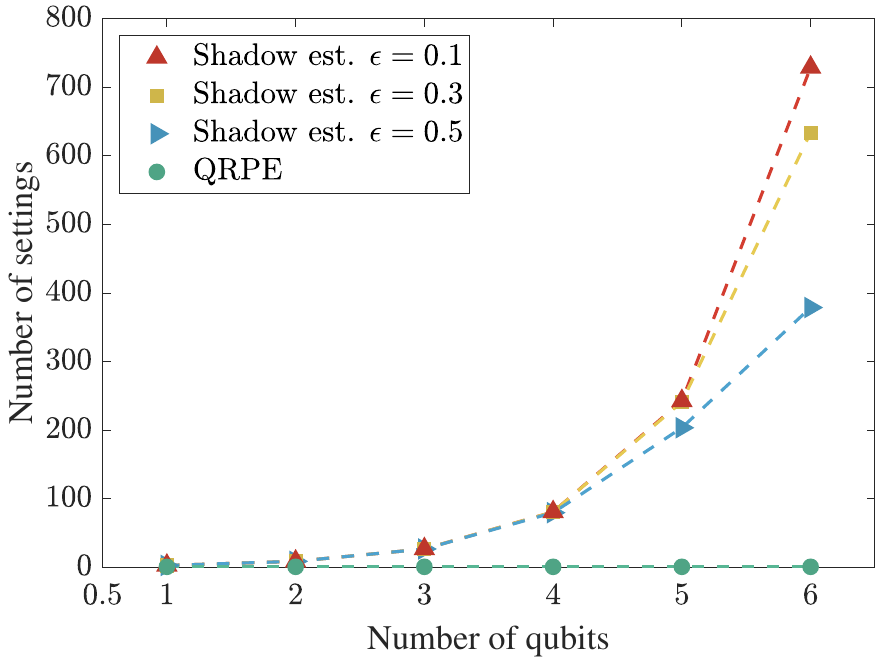}}
    \caption{Comparison with the shadow estimation with Pauli measurements~\cite{Huang2020}. (top) Average variance upper bound of fidelity estimation. For each system size we randomly generate 1000 pure states as the target states. The ratio of the average variance upper bounds for QRPE and shadow estimation remains almost constant as the number of qubits grows. (bottom) The average number of measurement settings invoked to reach a given precision with confidence level $0.90$. For shadow estimation, we assume the variance in single-snapshot estimator equals a tenth of the shadow norm on the top figure. While QRPE requires only a single setting, the standard shadow protocol with Pauli measurements~\cite{Huang2020,BiasedPauliShadowHadfield2022} requires exponentially increasing settings to reach an arbitrary precision. The reservoir setting we choose for qubit system is given by  $J = -0.41$, $P_1 = 4.0$, $P_2 = 1.3$, $E_1 = 0.71$, $E_2 = 0.46$ and $t = 1$. 
    The unit for the Plank's constant is meV$\cdot$ps.
    }
    \label{fig:RandFid}
\end{figure}

Theorem~\ref{thm:Scaling} indicates that for $k$-local tensor product property estimation, if a reservoir setting works well in the single-qubit case, then it also works well in the multi-qubit case. 
Thus, our approach for evaluating the reservoir parameters $J,\,P_{1,2},$ and $E_{1,2}$ is based on the performance in the single-qubit state overlap estimation task.
To see why overlap estimation reflects the overall performance of observable estimation, note that the single-snapshot estimator's variance of an arbitrary property $\mathcal{O}'$ satisfying $\mathcal{O}' = c_1\sigma + c_2\openone$ equals $c_1^2$ times that of $\sigma$, where $c_{1,2}$ are real numbers and $\sigma$ is a density matrix.
We find several settings that lead to small average variance upper bound of random target states, and choose one of them as the reservoir setting used in this manuscript. 

In the task of pure state fidelity estimation, the efficiency of the current reservoir estimator outperforms the random Pauli measurements only in a fraction of target pure states. We note that the optimal sample efficiency of shadow estimation with local measurements is achieved with random Pauli measurements~\cite{SdPOVM_PhysRevLett.129.220502,Huang2020}.
Also, the ratio of the average variance upper bound of shadow estimation and that of the QRPE scheme is slightly larger than unity and almost invariant with the system size. While the average efficiency of the current setting only marginally differ from that of the shadow estimation, the number of settings for the standard shadow estimation is exponentially large  compared to that of the present scheme. These results are shown in Fig.~\ref{fig:RandFid}. While one could obtain a scaling independent of the system size with global Clifford group measurements, it would require to implement complex Clifford compiling~\cite{DirectRB_PhysRevLett.123.030503, LowdepthShadow}.

Also, we could utilize random settings by probabilistic time multiplexing~(PTM), where the pair-wise training data $\bm{X}_\text{p}(t_k)$ are collected at $N_\text{time}$ different time points $\{t_k\,|\,k = 1,\,2,\,\dots,\,N_\text{time}\}$, and the single-snapshot estimator is constructed by measuring reservoir nodes at time $t_k$ with probability $p_k$, where
$\sum_k p_k = 1$. The optimization over probability distributions aims to reduce the average variance upper bound of single-qubit state overlap estimation, see Appendix~\ref{supp:B} for more details.

The QRPE protocol for linear functions is as follows:
\begin{enumerate}
  \item Perform a one-time estimation of the training states of a two-node reservoir and load the training data $\{\bm{X}_\text{p}\}$ to classical memory.
  \item Given properties $\{\mathcal{O}_i\}$, calculate weights $\{W_i\}$ with the training data. Obtain the worst-case variance upper bound $\max_i || \mathcal{B}_i||_\infty$.  Calculate $N_\text{sample}$ with the given confidence $1-\delta$ and additive error $\epsilon$. 
  \item Process i.i.d. copies of the unknown state $\sigma$ with the quantum reservoir network. Load $N_\text{sample}$ snapshots $\{X_i\}$ to classical memory.
  \item Calculate the estimated values $\{\tilde{\mathcal{O}}_i\}$ with the weights and snapshots.

\end{enumerate}
The reservoir estimators for nonlinear functions are based on U-statistics \cite{Huang2020,Hoeffding1948}, which is a generalization of the sample mean estimator.
It is worth noting that the reservoir snapshots are ready to be used to estimate future properties of the input state $\sigma$.

%%%%%%
\section{Applications}\label{Sec.III}
Here we provide a wide range of applications of the QRPE scheme. The reservoir initialization and evolution time are fixed for all applications. See Fig.~\ref{fig:RandFid} for the details of reservoir settings.
\subsection{Fidelity estimation}
Quantum fidelity is a widely used distance measure for quantum states, which is a linear function of the given state when the target state is a pure state~\cite{Fid2_PhysRevLett.106.230501, Fid_Cerezo2020}. 
In Fig.~\ref{fig:RandFid}, we present an illustrative example of how the reservoir estimation scheme works in fidelity estimation. Numerical results indicate that for the average worst-case variance upper bound for pure target states randomly generated from Haar measure, the ratio between QRPE and the random Pauli measurement is close to 1, while there exists an exponential improvement in the number of measurement settings.

\subsection{Entanglement detection}
Entanglement is an indispensable resource in tasks ranging from quantum computation to quantum communication~\cite{Ent_Vedral2014}. Whether one can claim its existence for a given system has aroused both theoretical and experimental interests~\cite{EntD1_PhysRevA.76.030305,EntD2_GUHNE20091,EntD3_Dimi2018,EntD4_PhysRevA.86.022311,EntD5_PhysRevA.72.022340}. However, the difficulty in describing the convex space of separable states poses a trade-off relation between the detection ability and  effectiveness of entanglement criteria~\cite{EntDectFL_PhysRevLett.129.230503}. 
With its capability of estimating multiple observables simultaneously, the QRPE scheme is a natural fit for the task of entanglement detection with linear entanglement criteria.

\begin{figure}[t]
    \includegraphics[width=.98\columnwidth]{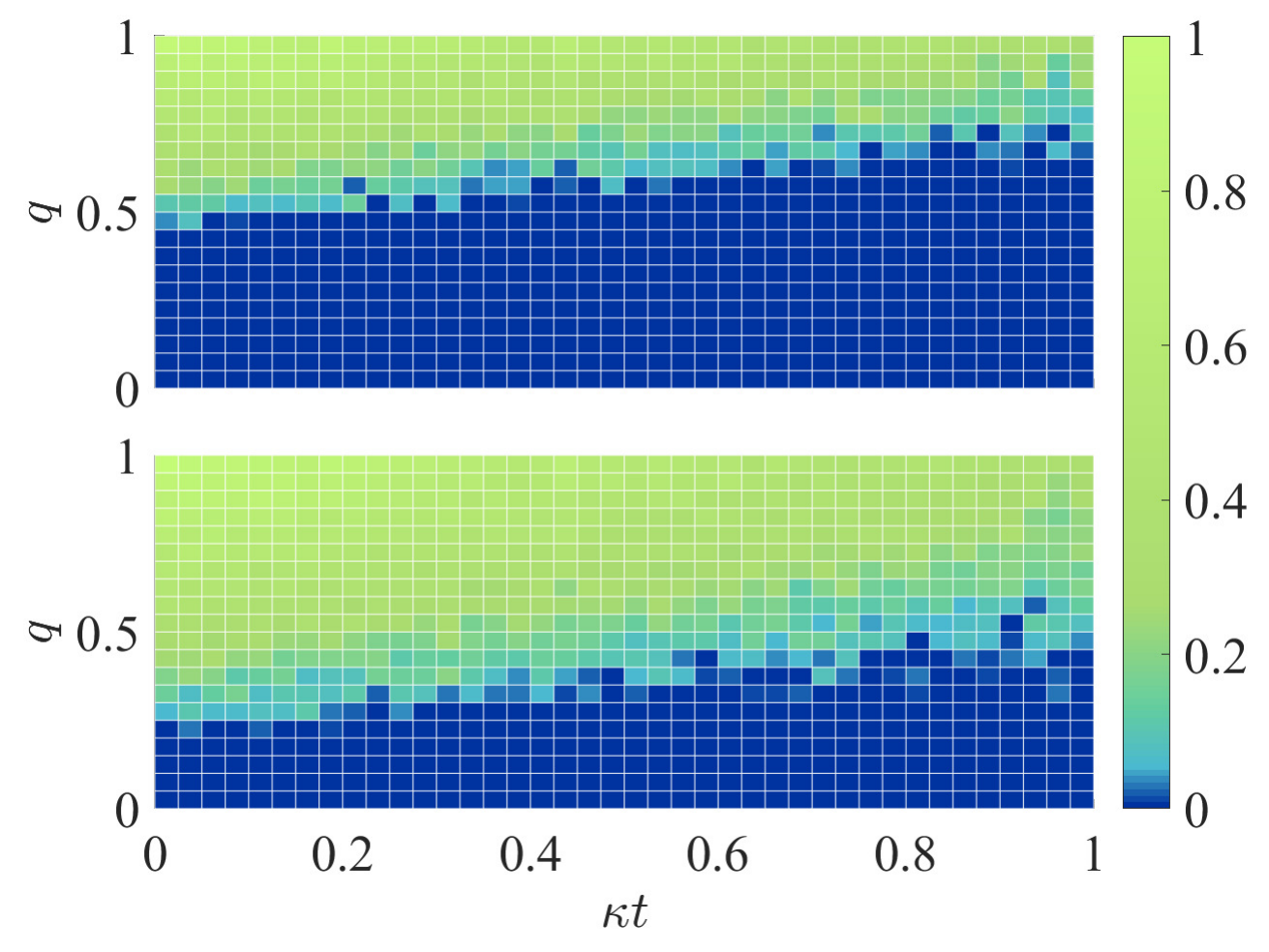}
    \caption{Observation of entanglement sudden death. (top): estimation of $W_\text{GME}$. (bottom): estimation of $W_\text{ME}$. We use 6000 measurements for each $\{q,\kappa t\}$, and the maximal estimation error for estimating GME and ME witnesses are 0.11 and 0.13 respectively. It can be observed that under a small dephasing noise, both the GME and ME witnesses vanish at a finite time for some initially entangled states, which indicates the emergence of ESD.
    }
    \label{fig:EntSD}
\end{figure}

We illustrate the QRPE scheme on the detection of an intriguing and more targeted phenomenon, the entanglement sudden death (ESD)~\cite{sdd_doi:10.1126/science.1167343,ESD_doi:10.1126/science.1139892}.  Consider a three-qubit GHZ-type state~\cite{sdd3_PhysRevA.79.012318}
\begin{equation}
    \rho = \frac{1-q}{8}\openone+q\rho_\text{G}\,,
\end{equation}
where $\rho_\text{G} = (\ket{000}+\ket{111})(\bra{000}+\bra{111})/2$. The dephasing channel is defined by the Kraus operators
\begin{equation}
    K_0 = \sqrt{1-p}\openone,\quad K_{1,2} = \frac{\sqrt{p}}{2}(\openone\pm \varsigma^z)\,.
\end{equation}
Set $p = 1-\exp{(-\kappa t)}$, then dephasing of the initial state $\rho$ reflects on a factor $\exp{(-\kappa t)}$ that times the off-diagonal elements.
In Fig.~\ref{fig:EntSD}, we use the following optimal linear entanglement witnesses for detecting genuine multipartite entanglement~(GME) and any multipartite entanglement~(ME) respectively~\cite{3QEW_doi.org/10.26421/QIC13.3-4}:
\begin{equation}
\begin{aligned}
        W_\text{GME} &= \openone - 2\rho_G\,,\\
        W_\text{ME} &= \openone - 4\rho_G + 2\rho_{G-}\,.
\end{aligned}
\end{equation}
where $\rho_{\text{G}-} = (\ket{000}-\ket{111})(\bra{000}-\bra{111})/2$. We normalize the spectral norm of the two witnesses, and estimate both of them simultaneously.
With a total of 6000 snapshots for each $\{q,\kappa t\}$, the maximal difference from the true value is 0.13. We see that under a small dephasing noise, both the GME and ME witnesses vanish at a finite time for some initially entangled states, indicating the emergence of ESD.
\subsection{Estimating expectation values of local and global observables}
Estimating the values of observables is a fundamental task in various quantum information processing tasks.
For local observables, we present the numerical results of estimating the local Pauli observables of randomly chosen four-qubit states.
The numerical result is shown in Fig.~\ref{fig:PaulikLocal}. We observe that the sample complexity required for achieving a given error rate of estimating Pauli observables increases exponentially with the locality, which agrees with Theorem.~\ref{thm:Scaling}.

\begin{figure}[t]
    \includegraphics[width=.98\columnwidth]{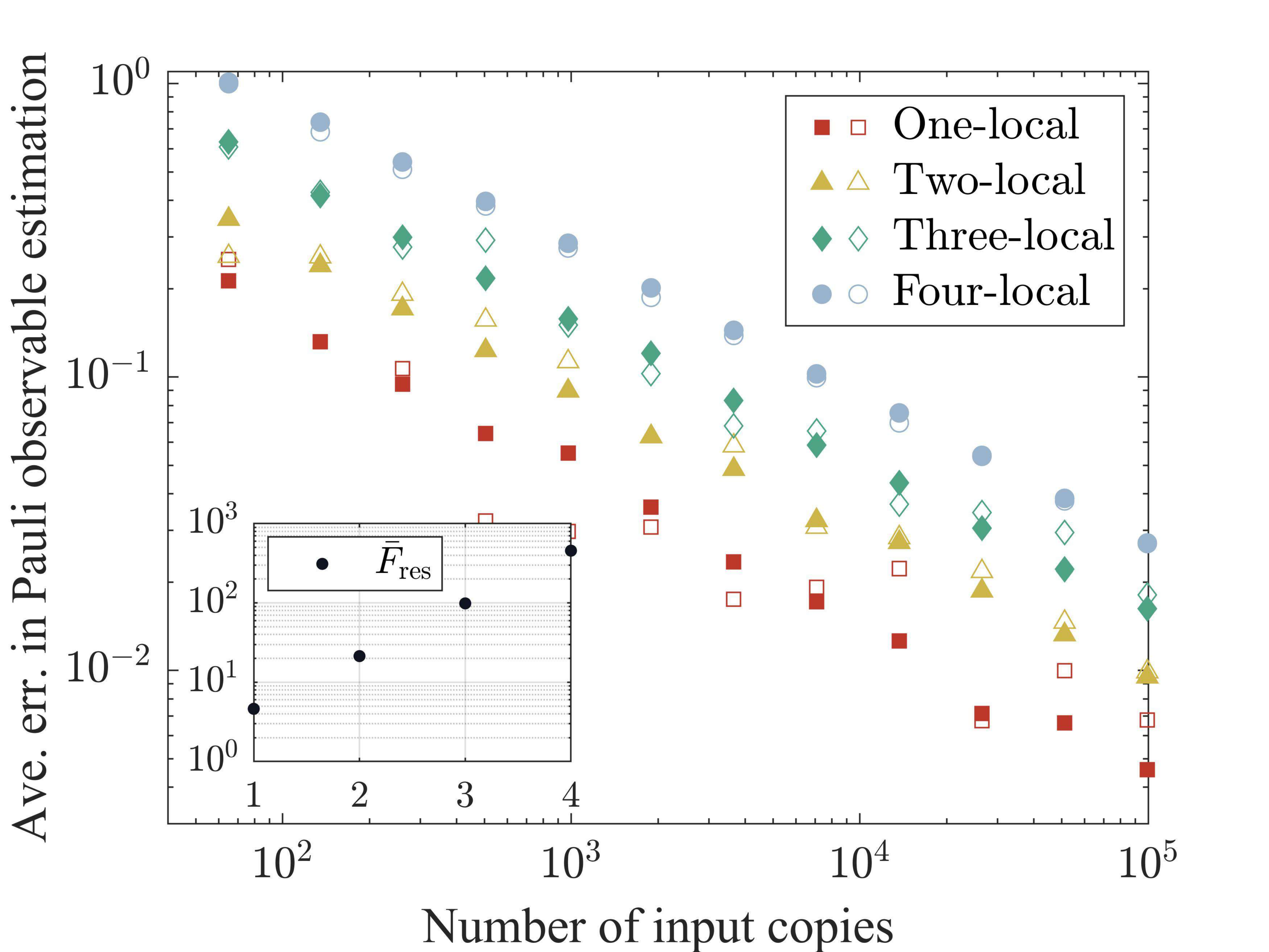}
    \caption{Numerical experiment for estimating Pauli observables of randomly chosen one, two, three and four-qubit input states. Each data point represents the average error of all the $k$-local Pauli observables of a random $k$-qubit input state. The solid marks represents the average performance of $30$ independent experiments, while the hollow marks represents a single experiment. The subfigure shows the variance upper bound $\bar{F}_\text{res}$, averaged over the $k$-local Pauli observables. Numerical results indicate that the worst-case sample complexity of estimating Pauli observables increases exponentially with the locality.
    }
    \label{fig:PaulikLocal}
\end{figure}

For global observables, we consider the task of GHZ state fidelity estimation. For an input GHZ state with three, six, nine and twelve qubits, we present the results of numerical experiment as the number of input copies versus the average error in Fig.~\ref{fig:GHZfid}. It can be seen that the sample complexity of identifying GHZ states with a given error rate increases exponentially with the number of qubits, which is similar to the case of local tensor product observables.

\begin{figure}[t]
    \includegraphics[width=.95\columnwidth]{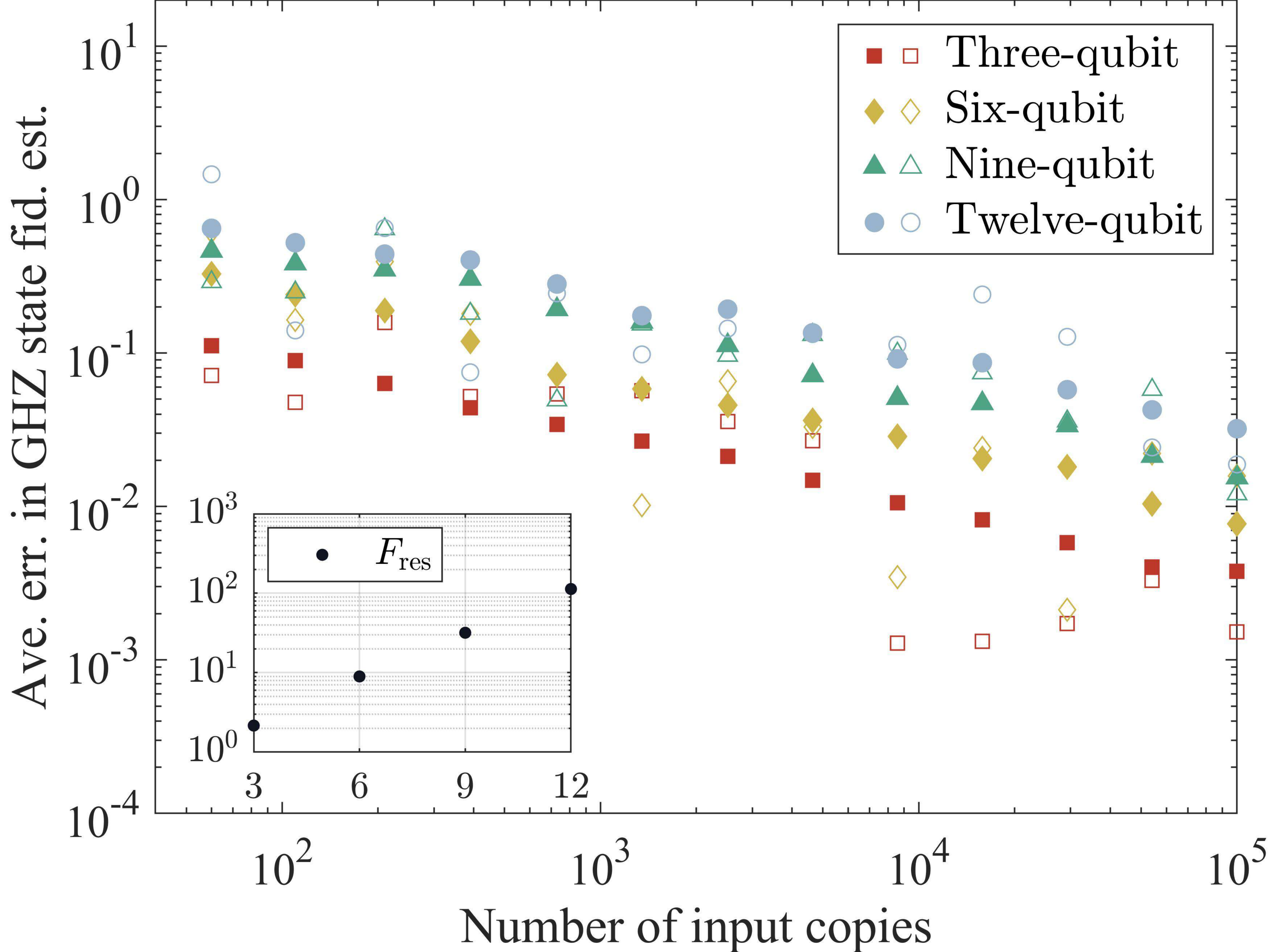}
    \caption{Numerical experiment for identifying GHZ states with fidelity estimation. The solid marks represents the average performance of $50$ independent experiments, while the hollow marks represents a single experiment. The subfigure shows the variance upper bound for the $k$-qubit GHZ state. Numerical results indicate that the sample complexity of identifying GHZ states increases exponentially with the number of qubits.
    }
    \label{fig:GHZfid}
\end{figure}
\subsection{Estimating nonlinear functions}
The reservoir measurements only describe linear functions in Eq.~\eqref{eq:unbiasedest}. Nevertheless, experimentally accessible nonlinear functions are typically measured by estimating linear functions that can be translated into the QRPE scheme.
For instance, the {\sc swap} trick is widely used in purity estimation. For two copies of the input state $\sigma$, the {\sc swap} operator $S$ obeys $\tr(S\sigma\otimes\sigma) = \tr(\sigma^2)$.
Numerical simulation of the purity estimation of 10000 random one-qubit input states shows that the average variance upper bound is around 2.9.

The R\'enyi entropy is another important nonlinear factor for characterizing entanglement, which is the logarithm of subsystem purity. 
The estimator of second R\'enyi entropy is
\begin{equation}
    \tr(\rho_A^2) = \tr(S_A\rho\otimes\rho)\,,
\end{equation}
where $S_A$ is the {\sc swap} operator acting on subsystem $A$ of two copies of $\rho$.
The second R\'enyi entropy of small subsystems is useful in avoiding weak barren plateaus (WBP)~\cite{WBP_PRXQuantum.3.020365}.
Here we perform WBP diagnosis with the reservoir estimation scheme. For an initial state $\ket{0}^{\otimes 14}$, each gate sequence of the variational quantum eigensolver~(VQE) circuit is composed of random local rotations $\exp(-\frac{i}{2}\theta \varsigma)$, where $\theta\in [\pi/20,\pi/20]$ and $\varsigma\in\{\varsigma_x, \varsigma_y, \varsigma_z\}$, and nearest-neighbor controlled-Z gates with periodic condition.
We detect the emergence of WBP by estimating the second R\'enyi entropy of a small region consisting of $N_A$ qubits via 2000 measurements at each circuit depth. The numerical experiment is shown in Fig.~\ref{fig:WBP}. We observe that the shaded region created by 10 independent estimations is centered by the real value of the second R\'enyi entropy with a small fluctuation, which grows with the size of the subsystem. Within a circuit depth of 100, the QRPE scheme efficiently detects the emergence of WBP by the clear phenomenon that the estimated value of second R\'enyi entropy approaches the page entropy.
 
\begin{figure}[t]
    \includegraphics[width=.95\columnwidth]{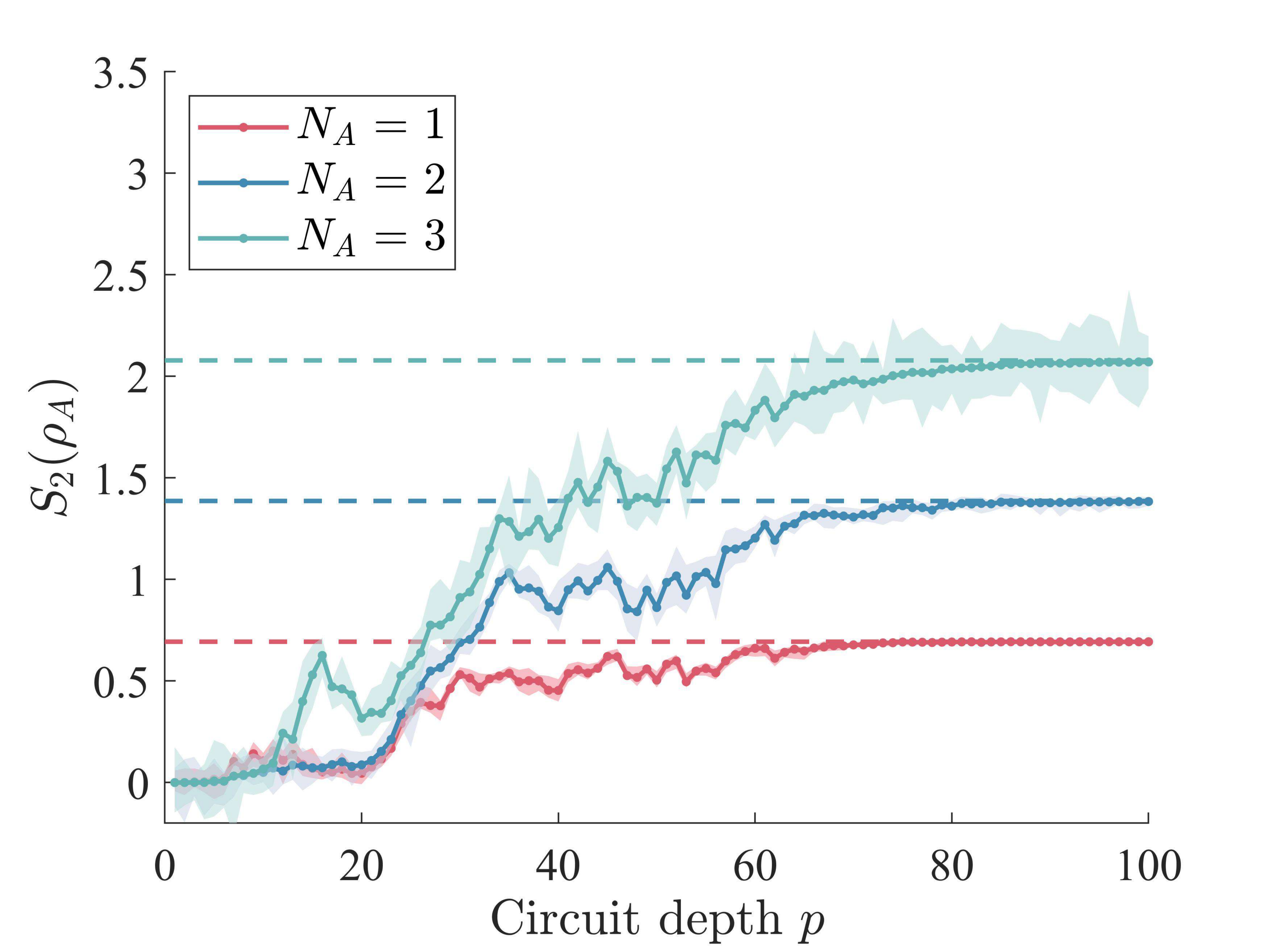}
    \caption{ Saturation of the second R\'enyi entropy $S_2(\rho_A)$ of a small region $A$ consisting of $N_A$ qubits. The solid curves represent $S_2(\rho_A)$ of a fourteen-qubit state generated by a VQE circuit with depth $p$, which approaches the page entropy represented by the dashed line. For reservoir estimation at each circuit depth we use $2000$ independent snapshots. The shaded region is the fluctuation of estimated value in 10 independent estimations. 
    }
    \label{fig:WBP}
\end{figure}

\subsection{Higher-Dimensional and Hybrid Systems}\label{sec:DV_QR}
Higher-dimensional systems and hybrid systems with non-identical local dimensions are of fundamental importance as a playground to reveal interesting quantum phenomena, such as quantum steering~\cite{Steer_RevModPhys.92.015001,Steer1_PhysRevA.94.062123} and the demonstration of contextuality and nonlocality trade-off~\cite{Hyb1_PhysRevLett.112.100401,Hyb2_PhysRevLett.116.090401}. Here we demonstrate the flexibility of our scheme beyond the qubit systems.

An intuitive generalization from the qubit case when the input state exists in a higher-dimensional or hybrid system is to design a quantum reservoir that has a compatible dimension. For example, if the input state consists of a qubit and a qutrit, then a quantum reservoir with two interacting pairs of qubits and qutrits can readily estimate properties of it.
The bosonic Hamiltonian for a pair of interacting reservoir nodes ${\langle i, j \rangle}$ is given by
\begin{equation}\label{eq: BosonH}
\begin{aligned}
    \hat{H}_{\langle i, j \rangle} & =  J \bigl( \hat{a}^\dagger_i \hat{a}_j + \hat{a}^\dagger_j \hat{a}_i \bigr) + P_1 (\hat{a}^\dagger_i+\hat{a}_i)\\
    & + P_2 (\hat{a}^\dagger_j+\hat{a}_j) + E_1 \hat{a}^\dagger_i \hat{a}_i + E_2 \hat{a}^\dagger_j \hat{a}_j \\
    & + \alpha_1 \hat{a}^\dagger_i \hat{a}^\dagger_i \hat{a}_i \hat{a}_i +\alpha_2 \hat{a}^\dagger_j \hat{a}^\dagger_j \hat{a}_j \hat{a}_j \,,
\end{aligned}
\end{equation}
where the operators $\hat{a}$ represent lowering operators of the quantum nodes (qudits). The parameters $J$, $P_{1,2}$, $E_{1,2}$ and $\alpha_{1,2}$ represent hopping, onsite driving field, onsite energy and nonlinear strength respectively. Similar to the qubit case, the reservoir map will be tomographically incomplete when $P_1$ and $P_2$ both equal 0 or $E_1$ and $E_2$ both equal 0.
For each pair of reservoir nodes ${\langle i, j \rangle}$, the $i$-th node is connected to the input system, and the number of pairs equals the number of constituents of the input state. We consider the unitary time evolution governed by the quantum Liouville equation.

For qudit systems with local dimension $d$, a superoperator basis consists of the generalized Gell-Mann matrices and the normalized identity~\cite{Gell_PhysRev.125.1067}. The readout operators are constructed in the same way with that of the qubit system, the only difference is that there are $d$ projection operators instead of two, corresponding to the $d$ possible population numbers.
Due to the tensor product structure of the measurements and training states, the reservoir estimation scheme can be naturally extended to hybrid systems with non-identical local dimensions, such as qubit-qutrit systems. The pair-wise reservoir setting for a pair of qudit nodes is the same with that of qudit reservoirs.

For application we apply virtual distillation~(VD)~\cite{Distill_PhysRevX.11.041036, Distill_https://doi.org/10.48550/arxiv.2203.07263,Distill_PRXQuantum.4.010303} to estimate the fidelity of a noisy state $\rho_\epsilon$,
\begin{equation}
    \langle\ket{\psi}\bra{\psi}\rangle_\text{VD} = \frac{\tr(\rho_\epsilon^m\ket{\psi}\bra{\psi})}{\tr(\rho_\epsilon^m )}\,,
\end{equation}
where 
\begin{equation}
    \rho_\epsilon = (1-\epsilon)\ket{\psi}\bra{\psi}+\epsilon\frac{\openone}{\tr(\openone)}\,.
\end{equation}
The estimator for $\rho_\epsilon^m$ is chosen as 
\begin{equation}
    \frac{1}{P (N, m)}\sum^* \prod_{i = 1}^m\hat{\rho}_{s_i}\,,
\end{equation}
where $\hat{\rho}_{s_i}$ is the single-snapshot estimator of the noisy state $\rho_\epsilon$, $P (N, m)$ is the $m$-permutations of $N$, $\sum^*$ denotes the summation over all distinct subscripts, i.e., ${s_1,s_2,\dots,\,s_m}$ is a $m$-tuple of indices from the set $\{1,\dots,\,N\}$ with distinct entries. The denominator is estimated similarly. The estimation results are illustrated in Fig.~\ref{fig:VD} for the following maximally entangled qubit-qutrit and two-qutrit pairs,
\begin{equation}
\begin{aligned}
    \ket{\psi_1} &= \frac{1}{2}(\ket{10}+\ket{12})+\frac{1}{\sqrt{2}}\ket{01}\,,\\
    \ket{\psi_2} &= \frac{1}{\sqrt{3}}(\ket{00}+\ket{11}+\ket{22})\,.
\end{aligned}
\end{equation}
We observe that the virtually distilled states exhibit a significantly improved fidelity as compared to the non-distilled states. Moreover, the statistical fluctuation for the $2\times 3$ system is smaller than that of the $3\times 3$ system, which is consistent with the performance in single-qudit pure state fidelity estimation. In the given reservoir setting for qutrit systems, the average ${F}_\text{res}$ for estimating fidelity of randomly chosen single-qutrit pure state is around 2.76, which is worse than the qubit case.

\begin{figure}[t]
    \includegraphics[width=.95\columnwidth]{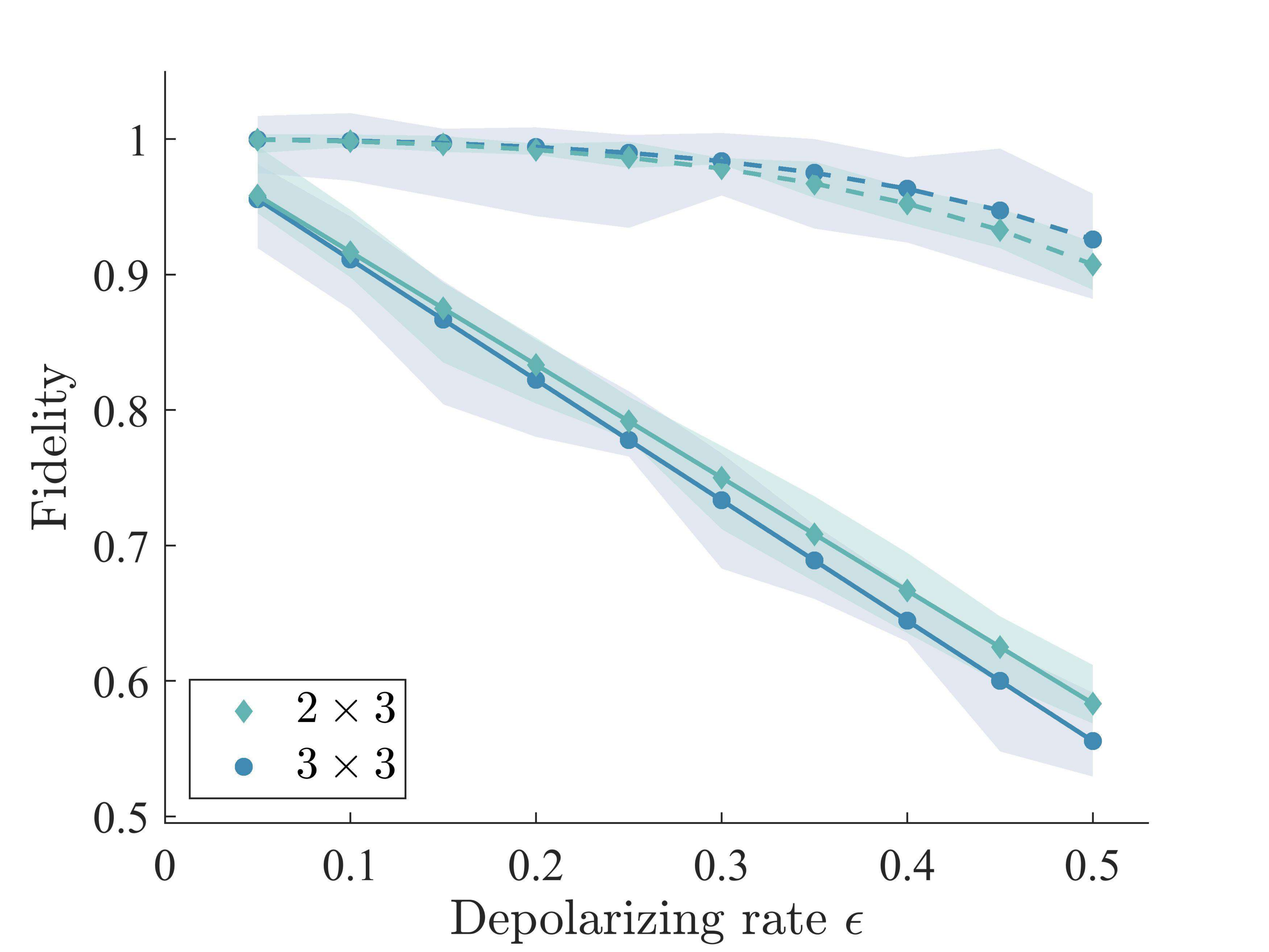}
    \caption{ The virtual distillation of  maximally entangled states mixed with depolarizing noise in $2\times 3$ and $3\times 3$ dimensional systems. The dashed curves represent the fidelity of the distilled states with $m = 2$, while the solid curves represent the fidelity of physical states with $m = 1$. The shaded region is the range of fluctuation of the estimated value in 10 independent estimations. For each reservoir estimation we use 10000 independent measurement readouts. The reservoir setting we choose for qutrit systems is $J = 0.9$, $P_1 = 2.1$, $P_2 = 1.1$, $E_1 = 1.1$, $E_2 = 0.4$, $\alpha_1 = 0.6$, $\alpha_2 = 0.7$.
    }
    \label{fig:VD}
\end{figure}

%%%%%%
\section{Continuous variable quantum reservoirs}\label{sec: CVres}
In this section we study a bosonic quantum reservoir that exists in continuous variable (CV) systems, while the input state emitted by a source still lies in discrete variable (DV) systems. The reservoir consists of identical pairs of interacting nodes, as shown in Fig.~\ref{fig: scheme}. A single pair of interacting reservoir nodes ${\langle k, l \rangle}$ is represented by the Hamiltonian
\begin{equation}
\begin{aligned}
    \hat{H}_{\langle k, l \rangle} = &\, J \bigl( \hat{a}^\dagger_k \hat{a}_l + \hat{a}^\dagger_l \hat{a}_k \bigr) -i P \sqrt{\gamma_k}(\hat{a}^\dagger_k- \hat{a}_k)\\
    & -i P \sqrt{\gamma_l}(\hat{a}^\dagger_l - \hat{a}_l) + E_1 \hat{a}^\dagger_k \hat{a}_k + E_2 \hat{a}^\dagger_l \hat{a}_l \\
    & + \alpha_1 \hat{a}^\dagger_k \hat{a}^\dagger_k \hat{a}_k \hat{a}_k +\alpha_2 \hat{a}^\dagger_l \hat{a}^\dagger_l \hat{a}_l \hat{a}_l \,,
\end{aligned}
\end{equation}
where $\hat{a}_i$ is the field operator of the $i$-th reservoir node, $P$ represents the strength of the coherent driving field, $E_i$ is the onsite energy and $\alpha$ is the onsite interaction strength (Kerr type nonlinearity).
We use essentially the same readout and training process as for the DV quantum reservoirs. 
Here the readout operators are projectors of Fock basis states. For instance, if the input state exists in a qubit system, then the measurements at each reservoir node include the first Fock state $\ket{0}\bra{0}$ and its complementary operator $\openone - \ket{0}\bra{0}$.

For a single pair of CV reservoir nodes, the input single-qudit system exists in a bosonic mode $\hat{b}_1$.  The density matrix $\rho$ of the entire system is in both of the quantum reservoir and incident modes, which is governed by the quantum master equation~\cite{PhysRevA.94.063825,QuantumNoise, Ghosh2019}
\begin{equation}\label{eq:Master}
    \begin{aligned}
    i\hbar\dot{\rho} =&\, [\hat{H}_{\langle 1, 2 \rangle}+f(t)\hat{H}_I, \rho] +\sum_{l = 1}^2 \frac{i\gamma_l}{2} \mathcal{L}(\hat{a}_l) \\
    &+ i f(t)  W^\text{in} \bigl([\hat{b}_1\rho, \hat{a}_{ 1}^\dagger]+ [\hat{a}_{ 1}, \rho\hat{b}_1^\dagger] \bigr)\\
     & +  f(t)\frac{i\eta}{2\gamma_1} \mathcal{L}(\hat{b}_1)\,,
    \end{aligned}
\end{equation}
where ${\hat{H}_I = \omega \hat{b}^\dagger_1 \hat{b}_1 - i P \sqrt{\eta/\gamma_1} (\hat{b}^\dagger_1 - \hat{b}_1)}$ is the input system Hamiltonian; ${\mathcal{L}(\hat{x}) = 2 \hat{x}\rho\hat{x}^\dagger - \hat{x}^\dagger\hat{x}\rho - \rho\hat{x}^\dagger\hat{x}}$ is the Lindblad operator acting on the field operator $\hat{x}$; $\gamma$ represents the decay rate; $W^\text{in}$ represents the input weight and $\eta = {(W^\text{in})}^2$ is set to remove the source photons that have excited the reservoir. On the right side of Eq.~\eqref{eq:Master}, the first term is the coherent Hamiltonian evolution of the reservoir; the second term represents the decay of reservoir modes; the second line represents the cascade coupling between the input mode $\hat{b}_1$ and the reservoir mode $\hat{a}_{1}$; and the last line represents the decay of the input mode. At time $t\in[t_1,t_1+\tau]$, ${f(t) = 1}$ and the input mode $\hat{b}_1$ is connected to the reservoir node. At other time points $f(t)$ equals zero. We measure the reservoir nodes at time ${t = t1+\tau+t2}$. The quantum master equation can be solved with the numerical method provided in Ref.~\cite{QRSP_PhysRevLett.123.260404}. Additionally, Eq.~\eqref{eq:Master} can be naturally generalized to the case where the input system has multiple constituents, each of which undergoes the same dynamical process. 

In numerical simulation, the truncated Fock space we use has a maximal photon number of 10. Here we set ${t_1 = 0.3}$, ${\tau = 1.4}$, ${t_2 = 0}$, ${\gamma_1 = 0.53}$, ${\gamma_2 = 0.87}$, ${P = 0.95}$, ${W^\text{in} = 0.56}$, ${E_1 = 0.94}$, ${E_2 = 0.57}$, ${J = -0.22}$ and ${\omega = 1}$. We choose ${\alpha_i/\gamma_i = 0.15}$, a ratio consistent with the experiment results~\cite{Delteil2019}. The spectral radius~\cite{QRSP_PhysRevLett.123.260404}, which is the largest eigenvalue of the hopping part of the reservoir Hamiltonian, is set to 1.5. With this reservoir setting, the average ${F}_\text{res}$ for estimating fidelity of randomly chosen single-qubit pure state is around 9.16. With numerical results, we find that the reservoir map becomes tomographically incomplete when the onsite driving field of the input system is missing. On the contrary, tomographic completeness can still be achieved if we set $E_i$, $\omega$ and $\alpha$ equal 0 or neglect the onsite driving field of the reservoir system. Also, we observe that the impact of Kerr nonlinearity on performance in numerical tests is not substantial.
Given the infinite dimensionality of the CV reservoir system, it has the capability to process any higher-dimensional or hybrid input states. However, with the same reservoir setting as described above, the average ${F}_\text{res}$ for estimating fidelity of randomly chosen single-qutrit pure states is around 490. Thus, similar to the DV quantum reservoir, the performance of CV quantum reservoir settings largely depends on the dimension of the input system.

\section{Conclusion}
We have presented a direct property estimation scheme, where a classical representation of quantum states is constructed with quantum reservoir processing and used for property estimation in the post-processing phase. Unlike existing techniques of shadow estimation that considers complex unitary ensembles and randomized measurements, our scheme explores the versatility of the quantum reservoir platform and requires only a single measurement setting.
The pair-wise interacting reservoir nodes considered in our scheme results in minimal quantum hardware for the DV reservoir, and minimal training cost that depends only on the input system.
The sample complexity has been rarely addressed in previous works on quantum reservoir computing.
In contrast, we have established a stringent performance guarantee regarding the number of samples to be processed by the reservoir network. Furthermore, our scheme can be naturally extended to higher-dimensional systems and hybrid systems with non-identical local dimensions. We complement the theoretical results with diverse applications.

For future research, reservoir computing is suitable for temporal pattern recognition, classification, and generation~\cite{review_Nakajima2020,Temporal1_Mujal2023,Review_TANAKA2019100}. While we have proposed a method for analyzing the statistical fluctuation of quantum reservoir outputs in property estimation, it is important to do so in temporal information processing tasks such as temporal quantum tomography~\cite{TemporalTomo_PhysRevLett.127.260401} and nonlinear temporal machine learning~\cite{Temporal_PhysRevApplied.8.024030}.  
Also, it is important to further devise an operational toolbox for optimizing the reservoir networks regarding the sample complexity and investigate whether the efficiency lower bound for local measurements~\cite{Huang2020,SdPOVM_PhysRevLett.129.220502} can be achieved on the quantum reservoir platform.
Moreover, there has been extensive research conducted on the learning abilities of quantum reservoir networks in tasks ranging from quantum to real-world problems~\cite{ComplexTask_Kutvonen2020,Xia2022,xia2023configured,Learning_PRXQuantum.3.030325}. These studies explore various network topologies and connectivities. While the pair-wise quantum reservoir construction is beneficial for property estimation, there lacks a general understanding of the role of topology and connectivity in quantum reservoir computing regarding sample complexity. 
Besides, quantum reservoir computing can be carried out with a single non-linear oscillator~\cite{PhysRevResearch.3.013077}, so it is important to analyze whether such a quantum reservoir can perform shadow estimation as well.
Another topic is the noise effect. There are inspiring discussions on the noisy training data~\cite{QELM}, robustness in tomographic completeness~\cite{TomoRub_PhysRevA.107.042402} and benefits of quantum noise~\cite{NoiseTem_PhysRevResearch.5.023057}. In our scheme the influence of time independent system noise on $\bm{\mathcal{T}}$ cancels out by a direct observation of Eq.~\eqref{eq:14}. However, a complete discussion on the noise effect and its mitigation in quantum reservoir processing is important for future experiments.

%%%%%%
\acknowledgments
We are grateful to Yanwu Gu, Zihao Li, Zhenpeng Xu and Ye-Chao Liu for helpful discussions. This work was supported by the National Natural Science Foundation of China (Grants No.~12175014 and No.~92265115) and the National Key R\&D Program of China (Grant No.~2022YFA1404900). S.G. acknowledges support from the Excellent Young Scientists Fund Program (Overseas) of China, and the National Natural Science Foundation of China (Grant No. 12274034).

%%%%%%
% \bibliographystyle{apsrev4-1}
% \bibliography{QRPE}


%%%%%%

%%%%%%
\onecolumngrid

\appendix
%%%%%%
\section{Theoretical interpretation of the QRPE scheme}\label{app:geointerp}
We start with the qubit input systems and reservoirs with qubit nodes for conciseness, but our method can be naturally extended to qudit input systems and hybrid input systems with different local dimensions.
To show how the QRPE scheme works, some may find it beneficial to examine the reservoir transformation from a geometric perspective.
The weights for an arbitrary property $\mathcal{O}$ and the readout probabilities for an arbitrary input state $\sigma$ are represented by vectors $W_\mathcal{O}$ and $\bar{X}_\sigma$ in the flat Euclidean space $\mathbf{E}^{d\times d}$, as shown in Eqs.~\eqref{eq:statemap} and \eqref{eq:obsmap} respectively. Since $\sum_{i = 1}^{d^2}{\langle \hat{o}_i\rangle_\sigma} = 1$, the set $\mathcal{C} = \{\bar{X}_\sigma\mid \sigma\succeq 0,\,\tr(\sigma) = 1\}$ is a region on a $(d^2-1)$-dimensional hyperplane. In addition, we have
\begin{lemma}\label{lemma:affinemap}
The reservoir transformation with pair-wise reservoir dynamics is an affine map from the ${(d^2-1)}$-dimensional state space to the ${(d^2-1)}$-dimensional convex region $\mathcal{C}$.
\end{lemma}
\begin{proof}
From Eq.~\eqref{eq:reservoirstate} we have $\tr\bigl(\hat{o}_i\rho(t)\bigr) = 
  \tr\bigl(\bra{0} \hat{U} (t)\hat{o}_i\hat{U}^\dagger (t) \ket{0}_\text{rest}\sigma\bigr)$, i.e.
\begin{equation}\label{eq:affinemap}
\begin{aligned}
  \langle \hat{o}_i\rangle_\sigma & = \llangle\mathcal{T}(\hat{o}_i)|\sigma\rrangle\,,
\end{aligned}
\end{equation}
where $\mathcal{T}(\hat{o}_i) = \bra{0} \hat{U} (t)\hat{o}_i\hat{U}^\dagger (t) \ket{0}_\text{rest}$, and  $|\cdot\rrangle$ is the Louiville superoperator representation, i.e., to vectorize an operator by expanding it on an orthonormal operator basis. Note that two quantum states can always be related by a CPTP map. Thus, there is a set of Kraus operators $\{K_k\}$ such that
\begin{equation}\label{eq:generalmap}
    \rho(t) = \sum_k K_k \sigma \otimes\ket{0}\bra{0}_\text{rest} K_k^\dagger\,,
\end{equation}
regardless of the method we apply to couple the input state to the reservoir or the reservoir dynamics. So, Eq.~\eqref{eq:affinemap} also holds in general QRP, where $\mathcal{T}(\hat{o}_i) = \bra{0} \sum_k K_k^\dagger \hat{o}_i K_k \ket{0}_\text{rest}$.
We observe that QRP is equivalent to a POVM measurement on the input state, with ancilla qubits prepared in the $\ket{0}$ state.
Define the dynamics matrix
\begin{equation}
    \bm{\mathcal{T}} = \bigl[\llangle\mathcal{T}(\hat{o}_1)|;\,\llangle\mathcal{T}(\hat{o}_2)|;\,\dots;\,\llangle\mathcal{T}(\hat{o}_{d^2})|\bigr]\,,
\end{equation}
then for an arbitrary input state $\sigma$,
\begin{equation}\label{eq:A4}
    \bar{X}_\sigma = \bm{\mathcal{T}}|\sigma\rrangle\,.
\end{equation}
We note that for pair-wise interaction, the columns of $\bm{\mathcal{T}}$ has a tensor product structure. Since the pair-wise evolution operator for $n$-qubit input states satisfies $\hat{U} (t) = \bigotimes_{i=1}^{n} \hat{U}_\text{p.w.} (t)$, we have
\begin{equation}
\begin{aligned}
    \mathcal{T}(\hat{C}_{S})& = \bra{0} \hat{U} (t)\hat{C}_{S}\hat{U}^\dagger (t) \ket{0}_\text{rest}
    = \bra{0}  \hat{U} (t)
    \prod_{i\in S}  \frac{\openone-\varsigma_{i}^z}{2}   \prod_{j\notin S}  \frac{\openone+\varsigma_{j}^z}{2}   \hat{U}^\dagger (t)
    \ket{0}_\text{rest} = \bigotimes_{i=1}^n \mathcal{T}_i(\hat{C}_{S})\,,
\end{aligned}
\end{equation}
where 
\begin{equation}
    \mathcal{T}_i(\hat{C}_{S}) = \bra{0} \hat{U}_\text{p.w.} (t) \mathcal{M}_i \hat{U}_\text{p.w.}^\dagger (t) \ket{0}
\end{equation}
is a $2\times 2$ matrix, $\mathcal{M}_i$ is one of the two-qubit Pauli-$Z$ projection operators acting on the pair of interacting reservoir nodes that connect to the $i$-th qubit of the input state, and $\bra{0}\cdot\ket{0}$ acts on the reservoir node that doesn't connect to the input state directly. There are only four distinct elements $\{\mathcal{T}_1, \mathcal{T}_2, \mathcal{T}_3, \mathcal{T}_4\}$ in the set $\{\mathcal{T}_i(\hat{o}_1),\,\dots,\,\mathcal{T}_i(\hat{o}_{d^2})\}$, which correspond to the four different two-qubit Pauli-$Z$ projections $\{\mathcal{M}_i\}$. In fact, 
\begin{equation}\label{eq: ApwTM}
    \bm{\mathcal{T}} = \bigotimes_{i=1}^{n} \bm{\mathcal{T}}_\text{p.w.}\,,
\end{equation}
where $\bm{\mathcal{T}}_\text{p.w.}$ is the pair-wise dynamics matrix for a single qubit input state, i.e., $\bm{\mathcal{T}}_\text{p.w.} = \bigl[\llangle\mathcal{T}_1|;\,\llangle\mathcal{T}_2|;\,\llangle\mathcal{T}_3|;\,\llangle\mathcal{T}_4|\bigr]$.
For general QRP in Eq.~\eqref{eq:generalmap}, the tensor product structure of training data $\bm{\mathcal{T}}$ persists if the coupling of the input state to the reservoir and the reservoir dynamics are \emph{local}, i.e., each constituent of the input system undergoes independent dynamical processes.
The training process ascertains that the set of operators $\{\mathcal{T}_1, \mathcal{T}_2, \mathcal{T}_3, \mathcal{T}_4\}$ is tomographically complete, i.e., $\bm{\mathcal{T}}_\text{p.w.}$ is invertible, then $\bm{\mathcal{T}}$ is also invertible.
So the reservoir transformation is an affine map from the state space to the $(d^2-1)$-dimensional region $\mathcal{C}$. Also, we note that Eq.~\eqref{eq:A4} indicates
\begin{equation}\label{eq:A7}
    \lambda \bar{X}_{\sigma_0} + (1-\lambda) \bar{X}_{\sigma_1} = \bar{X}_{\lambda {\sigma_0} + (1-\lambda) {\sigma_1}}\,,
\end{equation}
where $1\ge\lambda\ge0$. Since $\bar{X}_{\lambda {\sigma_0} + (1-\lambda) {\sigma_1}}\in\mathcal{C}$, $\mathcal{C}$ is a convex region.
\end{proof}

The basic QRPE model behind various scenarios that we consider in this work is as follows:
\begin{theorem}\label{thm:scheme}
For an unknown state $\sigma$ in a $d$-dimensional Hilbert space, one can linearly combine the readout operators $\{\hat{o}_i\}$ of an ${N_\text{node} = 2\log_2(d)}$ reservoir into unbiased estimators of arbitrary properties $\{\mathcal{O}_i\}$, if supplemented with the training data from $d^2$ linearly independent states.  
\end{theorem}
\begin{proof}
Firstly we note that the $d^2$ training points $\bigl\{\bar{X}_{\ket{\varphi_k}\bra{\varphi_k}}\bigr\}$ are not confined to any $(d^2-2)$-dimensional subspace, i.e., they form a $(d^2-1)$-simplex~\cite{Geometry}.
Linear independence indicates that the training states in Eq.~\eqref{eq:trainingstate} form a $(d^2-1)$-simplex in the state space, then Lemma~\ref{lemma:affinemap} shows that the training points also form a $(d^2-1)$-simplex in $\mathcal{C}$.
Thus, an arbitrary $d^2$-dimensional target vector $Y^\text{tar}$ in Eq.~\eqref{eq:tragetvector} can be realized by taking the inner product of the training vectors $\bigl\{\bar{X}_{\ket{\varphi_k}\bra{\varphi_k}}\bigr\}$ with a unique weight vector $W$, which is acquired in the training phase.
Also, one gets
\begin{equation}
    \bar{X}_\sigma = \sum_{k = 1}^{d^2} c_k^\sigma \bar{X}_{\ket{\varphi_k}\bra{\varphi_k}}\,,
\end{equation}
where $\{c_k^\sigma\}$ is the set of unique barycentric coordinates for $\bar{X}_\sigma$. 
We have
\begin{equation}\label{eq:proofofUbsdEst}
\begin{aligned}
    W\cdot \bar{X}_\sigma &= \sum_{k = 1}^{d^2} c_k^\sigma W\cdot \bar{X}_{\ket{\varphi_k}\bra{\varphi_k}} = \tr\Bigl(\mathcal{O}\sum_{k = 1}^{d^2} c_k^\sigma \ket{\varphi_k}\bra{\varphi_k}\Bigr) = \tr(\mathcal{O}\sigma)\,,
\end{aligned}
\end{equation}
which proves that $\Omega$ in Eq.~\eqref{eq:estimator} is an unbiased estimator of $\mathcal{O}$. Note that Lemma~\ref{lemma:affinemap} draws forth the relation
\begin{equation}
    \bar{X}_\sigma = \sum_{k = 1}^{d^2} c_k^\sigma \bar{X}_{\ket{\varphi_k}\bra{\varphi_k}} \Leftrightarrow \sigma = \sum_{k = 1}^{d^2} c_k^\sigma \ket{\varphi_k}\bra{\varphi_k}\,,
\end{equation}
then the last equality in Eq.~\eqref{eq:proofofUbsdEst} follows.
\end{proof}
The same results can be drawn with linear algebra. In the training phase, the $k$-th element of the target vector is
\begin{equation}
    Y^\text{tar}_k = \llangle\mathcal{O}|\varrho_k\rrangle\,,
\end{equation}
where $\varrho_k = \ket{\varphi_k}\bra{\varphi_k}$ is the density matrix of the $k$-th training state. Denote
\begin{equation}
    \bm{M}_\text{t} = \bigl[|\varrho_1\rrangle,\,|\varrho_2\rrangle,\,\dots,\,|\varrho_{d^2}\rrangle\bigr]\,,
\end{equation}
which is an invertible matrix, then
\begin{equation}
    Y^\text{tar} = \llangle\mathcal{O}| \bm{M}_\text{t}\,.
\end{equation}
Note that the matrix of training data satisfies
\begin{equation}
\begin{aligned}
      \bm{X}_\text{t} &= \bigl[\bar{X}_{\varrho_1},\,\bar{X}_{\varrho_2},\,\dots,\,\bar{X}_{\varrho_{d^2}}\bigr]
      = \bm{\mathcal{T}} \bm{M}_\text{t}\,.
\end{aligned}
\end{equation}
Thus, the weight vector is
\begin{equation}\label{eq:suppWeight}
\begin{aligned}
W &= Y^\text{tar}{\bm{X}_\text{t}}^{-1}
= \llangle\mathcal{O}|\bm{\mathcal{T}}^{-1}\,,
\end{aligned}
\end{equation}
where
\begin{equation}\label{eq:A16}
    \bm{\mathcal{T}} = {\bm{X}_\text{t}}\bm{M}_\text{t}^{-1}\,.
\end{equation}
Finally, we obtain
\begin{equation}\label{eq:A17}
    W \bar{X}_\sigma = \llangle\mathcal{O}|\sigma\rrangle = \tr(\mathcal{O}\sigma)\,.
\end{equation}

In the training of pair-wise interacting reservoirs, we do not require the entire matrix $\bm{\mathcal{T}}$, instead, we obtain the pair-wise matrix $\bm{\mathcal{T}}_\text{p.w.}$, and decompose the $n$-qubit observable $\mathcal{O}$ into a tensor product structure:
\begin{equation}
    \mathcal{O} = \sum_i \bigotimes_{j=1}^n \mathcal{O}_{i,j}\,,
\end{equation}
then the weight vector is
\begin{equation}
\begin{aligned}
W &= \llangle\mathcal{O}|\bm{\mathcal{T}}^{-1} =  \sum_i \bigotimes_j \mathcal{W}_{i,j} \,,
\end{aligned}
\end{equation}
where $ \mathcal{W}_{i,j} = \llangle\mathcal{O}_{i,j}| \bm{\mathcal{T}}_\text{p.w.}^{-1}$.
To reduce the complexity in classical post-processing, one only stores the set of two-node weights $\{\mathcal{W}_{i,j}\}$. Note that the single copy readout $X$ shares the same tensor product structure, $X = \bigotimes_j \mathcal{X}_j$, where $\mathcal{X}_j$ is a $4\times 1$ dimensional vector. 
Denote
\begin{equation}
    \hat{\Omega}_{i,j} = \mathcal{W}_{i,j}\cdot \mathcal{X}_j\,,
\end{equation}
then the single-copy estimator is 
\begin{equation}
    \hat{\Omega} = \sum_i \prod_j \hat{\Omega}_{i,j}\,.
\end{equation}
Thus, the classical storage consumption, i.e., $\{\mathcal{W}_{i,j}\}$ and $\{\mathcal{X}_j\}$, scales linearly with the number of qubits for observables that can be expressed as low-rank tensors.

\section{Efficiency for estimating linear functions}\label{supp:B}
Here we provide a method for analyzing the efficiency of QRPE in the worst-case scenario, which relies on the prior knowledge of the training data and the target properties.

\subsection{Single-snapshot estimator}
For a single input copy $\sigma$, the corresponding readouts can be combined into unbiased estimators of compatible properties. The efficiency of the QRPE scheme is characterized by the variance of the estimators.

\begin{lemma}\label{lemma: Var}
For a property $\mathcal{O}$, construct the single-snapshot estimator $\hat{\Omega}$ with the training data. Then in the worst-case scenario, we have
\begin{equation}
    \max_{\sigma}\text{Var}(\hat{\Omega}) \le F_\text{res}({\mathcal{O}},{\bm{X}_\text{t}})\,,
\end{equation}
where the function $F_\text{res}$ depends on the training data:
\begin{equation}
    F_\text{res}(\mathcal{O},{\bm{X}_\text{t}}) = \max_{\substack{\sigma\succeq 0\\ \tr(\sigma) = 1}} W\odot W \bar{X}\,.
\end{equation}
\end{lemma}
\begin{proof}
The single-snapshot estimator $\hat{\Omega}$ is a linear combination of the measurement results of mutually exclusive readout operators $\{\hat{o}_i\}$ with the corresponding weights $\{w_i\}$, i.e.,
\begin{equation}\label{eq:estimator}
    \hat{\Omega} \equiv  W\cdot \hat{X} = \sum_{i = 1}^{N_\text{read}}{w_i \hat{x}_i}\,.
\end{equation}
Thus, the expectation of $\hat{\Omega}$ is the inner product between $W$ and $\bar{X}$. 
The variance of $\hat{\Omega}$ is
\begin{equation}%\label{appeq: Var}
  \begin{aligned}
    \text{Var}(\hat{\Omega}) &= \text{Var}\Bigl(\sum_i{w_i\hat{x}_i}\Bigr) = W\odot W \bar{X} - (W \bar{X})^2\,,
  \end{aligned}
\end{equation}
where $\odot$ represents the Hadamard product.
Thus,
\begin{equation}\label{eq: Varupbdopt}
\begin{aligned}
    \max_{\bar{X}\in\mathcal{C}}{\text{Var}(\hat{\Omega})} &\le  \max_{\substack{\sigma\succeq 0\\ \tr(\sigma) = 1}} W\odot W \bar{X}\,,
\end{aligned}
\end{equation}
where we have neglected the term $-({W \bar{X}})^2$.
The inequality saturates when $\tr(\mathcal{O}\sigma) = 0$.
\end{proof}
Since we have $W = \llangle\mathcal{O}|{\bm{\mathcal{T}}}^{-1}$ and $X = {\bm{\mathcal{T}}}|\sigma\rrangle$, we can directly compute the factor ${F}_\text{res}$.
Define an operator $\mathcal{B}$ such that
\begin{equation}
\llangle\mathcal{B}|{\bm{\mathcal{T}}}^{-1} = \llangle\mathcal{O}|{\bm{\mathcal{T}}}^{-1} \odot \llangle\mathcal{O}|{\bm{\mathcal{T}}}^{-1}\,,
\end{equation}
then 
\begin{equation}
    {F}_\text{res} = \max_{\substack{\sigma\succeq 0\\ \tr(\sigma) = 1}}\tr(\mathcal{B}\sigma) = || \mathcal{B}||_\infty
\end{equation}
where $||\cdot||_\infty$ represents the spectral norm.
Also, it is worth noting that estimating identity operator doesn't affect the variance. Define $\hat{\Omega}' = (W+c W_{\openone}) \hat{X}$,
where $W_{\openone} = [1,1,\dots,1]$ is the weight operator for the identity observable, and $c\in \mathbb{R}$.
Since $c_1 W_{\mathcal{O}_1} + c_2 W_{\mathcal{O}_1} = 
W_{c_1 \mathcal{O}_1+ c_2 \mathcal{O}_2}$,
$\hat{\Omega}'$ is the estimator for $\mathcal{O}' = \mathcal{O} + c\openone$. 
Note that $\text{Var}(\hat{\Omega}) = \mathbb{E}[(\hat{\Omega} - \mathbb{E}(\hat{\Omega}))^2]$,
and by construction $\sum_i \hat{x}_i = 1$,
we have $\text{Var}(\hat{\Omega}') = \text{Var}(\hat{\Omega})$. 
We observe that $\hat{\Omega}_0$ results in a tighter worst-case variance upper bound on average in fidelity estimation of random pure target states.
Also, for properties that only differs by a product factor $c$,  \begin{equation}
     F_\text{res}(c\mathcal{O},{\bm{X}_\text{t}}) = c^2 F_\text{res}(\mathcal{O},{\bm{X}_\text{t}})\,.
\end{equation}
With these properties, we note that if a reservoir setting works relatively better for overlap estimation of a target state $\sigma$, one could expect that it works better for estimating arbitrary linear properties that can be decomposed as
\begin{equation}\label{eq: LinearObsRelation}
    \mathcal{O} = c_1\sigma + c_2 \openone\,,\quad c_{1,2}\in \mathbb{R}\,.
\end{equation}
An arbitrary normal matrix is related to a density matrix by
Eq.~\eqref{eq: LinearObsRelation}.
Thus, a well trained reservoir setting for overlap estimation of random target states is applicable to the future estimation of multiple unknown normal matrices. 
In order to analyze the search of good reservoir parameters, we fix the reservoir parameters $\{J, P_i, E_i\}$ and check the dependency of reservoir performance on evolution time. For the evolution time from $t = 0$ to $t = 5.5$, the variance upper bound of single-qubit state fidelity estimation averaged for 300 random target pure states varies, as shown in Fig.~\ref{fig:Var_T}.
It can be observed that the average variance upper bound is stable around some local minimums in a short time window.

\begin{figure}[t]
    \includegraphics[width=.7\columnwidth]{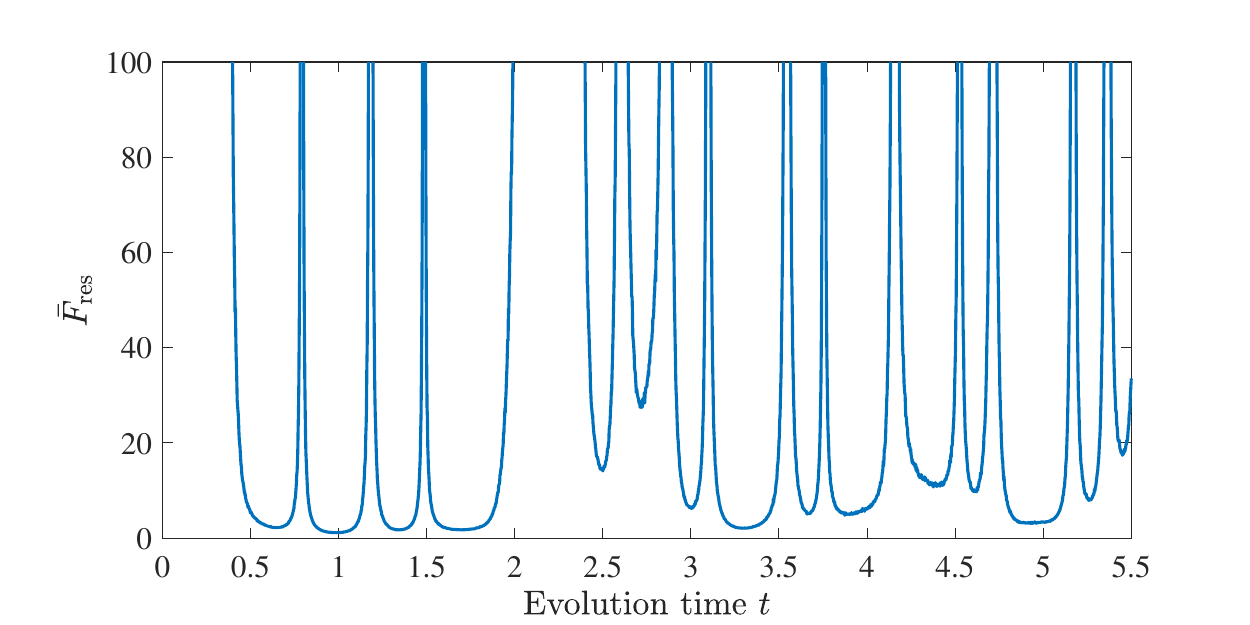}
    \caption{ Average variance upper bound of single-qubit state fidelity estimation that varies with evolution time from $t = 0$ to $t = 5.5$. The parameters $J, P_i, E_i$ are the same with main text. At each evolution time, the variance upper bound for pure state fidelity estimation $\bar{F}_\text{res}$ is averaged over 1000 random target states. While there are multiple local minimums with different slopes at neighboring points, the performance of reservoir parameters remains stable for a short time period at some local minimums.
    }
    \label{fig:Var_T}
\end{figure}

Furthermore, we could analyze the efficiency scaling of estimating $k$-local tensor product observables with the help of pair-wise reservoir dynamics. Here we present the detailed proof for Theorem 2.
\begin{lemma}[Theorem 2]
    For a $k$-local tensor product property, e.g.
    \begin{equation}
    \mathcal{O} = \bigotimes_{i = 1}^k \mathcal{O}_i \bigotimes_{i = k+1}^n \openone_i\,.
\end{equation}
The worst-case variance upper bound is the product of that of each local properties, i.e.,
\begin{equation}
    || \mathcal{B}||_\infty 
    = \prod_{i = 1}^k || \mathcal{B}_i||_\infty\,.
\end{equation}
\end{lemma}
\begin{proof}
    Suppose the property $\mathcal{O}$ has $k$-local tensor product structure, without loss of generality, we have
\begin{equation}
    \mathcal{O} = \bigotimes_{i = 1}^k \mathcal{O}_i \bigotimes_{i = k+1}^n \openone_i\,.
\end{equation}
Note that $W = \llangle\mathcal{O}|{\bm{\mathcal{T}}}^{-1}$, $\bar{X} = \bm{\mathcal{T}}|\sigma\rrangle$ and $\bm{\mathcal{T}} = \bigotimes_{i = 1}^n \bm{\mathcal{T}}_\text{p}$, there is
\begin{equation}
    \begin{aligned}
    \max_{\substack{\sigma\succeq 0\\ \tr(\sigma) = 1}} W\odot W \bar{X} & = \max_{\substack{\sigma\succeq 0\\ \tr(\sigma) = 1}} \bigotimes_{i = 1}^k  \llangle\mathcal{O}_i|{\bm{\mathcal{T}}_\text{p}}^{-1} \odot \llangle\mathcal{O}_i|{\bm{\mathcal{T}}_\text{p}}^{-1} \bigotimes_{i = k+1}^n W_{\openone_i} \bm{\mathcal{T}}|\sigma\rrangle\,.
    \end{aligned}
\end{equation}
Suppose the observable corresponding to the weight operator $\llangle\mathcal{O}_i|{\bm{\mathcal{T}}_\text{p}}^{-1} \odot \llangle\mathcal{O}_i|{\bm{\mathcal{T}}_\text{p}}^{-1}$ is $\mathcal{B}_i$, i.e.,
\begin{equation}
    \llangle\mathcal{O}_i|{\bm{\mathcal{T}}_\text{p}}^{-1} \odot \llangle\mathcal{O}_i|{\bm{\mathcal{T}}_\text{p}}^{-1} = \llangle\mathcal{B}_i|{\bm{\mathcal{T}}_\text{p}}^{-1}\,,
\end{equation}
then
\begin{equation}
    \begin{aligned}
    \max_{\substack{\sigma\succeq 0\\ \tr(\sigma) = 1}} W\odot W \bar{X} 
    = \max_{\substack{\sigma\succeq 0\\ \tr(\sigma) = 1}}  \tr\bigl(\bigotimes_{i = 1}^k \mathcal{B}_i \bigotimes_{i = k+1}^n \openone_i \sigma \bigr) = \prod_{i = 1}^k || \mathcal{B}_i||_\infty\,,
    \end{aligned}
\end{equation}
where $\max_{\substack{\sigma\succeq 0\\ \tr(\sigma) = 1}} W\odot W \bar{X}$ is the worst-case variance upper bound for the single-snapshot estimator.
\end{proof}

\subsection{Multi-snapshot estimator}\label{supp:MoM}
The variance can be further suppressed by combining single-snapshot estimators with statistical methods. The median of means method is expected to reduce the effect of outliers~\cite{Huang2020,RSPRXQuantum.2.030348}, the efficiency of which is given by the following theorem.
\begin{lemma}\label{lemma:MOM}
For target properties $\{\mathcal{O}_i| i = 1,\,2,\,\dots,\,M\}$ and training data ${\bm{X}_\text{t}}$, set 
\begin{equation}
    N_\text{sample} = {68}/{\epsilon^2} \ln(2M/\delta)\max_{i} F_\text{res}({\mathcal{O}_i},{\bm{X}_\text{t}})\,.
\end{equation}
Then, consuming $N_\text{sample}$ i.i.d. copies of input states suffice to construct median of means estimators $\{\hat{\Omega}_{\text{MoM},i}\}$, satisfying
\begin{equation}
    \bigl|\hat{\Omega}_{\text{MoM},i} - \tr(\mathcal{O}_i\sigma)\bigr|\le\epsilon\quad\text{for}\quad i = 1,\,2,\dots,\,M\,,
\end{equation}
with probability no less than $1-\delta$.
\end{lemma}
\begin{proof}
Divide the $N_\text{sample}$ readouts into $K$ batches, where
\begin{equation}
    K = 2\ln(2M/\delta)\,.
\end{equation}
Compute the sample mean estimator of each patch:
\begin{equation}
\begin{aligned}
    \hat{\Omega}_{\text{M},i}^{(k)} &= \frac{K}{N_\text{sample}}\sum_{j = 1}^N W_i X_j^{(k)}\,,\quad k = 1,\,2,\,\dots,\, K,
\end{aligned}
\end{equation}
then the median of means estimator is 
\begin{equation}
    \hat{\Omega}_{\text{MoM},i} = \text{Median}( \hat{\Omega}_{\text{M},i}^{(1)},\,\hat{\Omega}_{\text{M},i}^{(2)},\,\dots,\, \hat{\Omega}_{\text{M},i}^{(K)})\,.
\end{equation}
Lemma.~\ref{lemma: Var} and the median of means method~\cite{Huang2020} guarantees that 
\begin{equation}
    \text{Pr}\bigl(|\hat{\Omega}_{\text{MoM},i} - \tr(\mathcal{O}_i\sigma)|>\epsilon\bigr)<\frac{\delta}{M}\,.
\end{equation}
Finally, the union bound shows that the overall confidence level to estimate $M$ properties is no less than $1-\delta$.
\end{proof}

\subsection{Variance reduction via probabilistic time multiplexing}
The dynamic nature of quantum reservoirs is exploited with probabilistic time multiplexing.
For each input copy of the unknown state, the evolution time of the reservoir is randomly chosen from the set $\{t_1,\,t_2,\,\dots,\,t_{N_\text{time}}\}$ with a probability distribution $\{p_1,\,p_2,\,\dots,\,p_{N_\text{time}}\}$. The probability vector of readout operators after probabilistic time multiplexing (PTM) is 
\begin{equation}
    \bar{X}^{\text{ptm}} = \sum_{i = 1}^{N_\text{time}}{p_i \bar{X}(t_i)}\,.
\end{equation}
Solve the following optimization
\begin{equation}
\begin{aligned}
    F(\mathcal{O},{\bm{X}_\text{t}^{\text{opt}}})
    = \min_{\{p_i\}}F(\mathcal{O},{\bm{X}_\text{t}}^{\text{ptm}})\,,
\end{aligned}
\end{equation}
we are likely to reach a better variance upper bound.
Note that PTM does not require changing either the initialization parameters of the reservoir or the measurement setting.
Since the PTM is a one-time optimization, in this work we search the good probability distributions with brute force, which is achieved by the following algorithm.

{\centering
\begin{minipage}{.6\linewidth}
    \begin{algorithm}[H]\label{algo: PTM}
    \caption{{\bf PTM}}
    \begin{algorithmic}
    \vspace*{0.1cm}
    \State Given properties $\{\mathcal{O}_i\,|\,1,\,2,\,\dots,\,M\}$.
    \State Load training data $\{\bm{X}_\text{t}(t_i)\,|\,i = 1,\,2,\,\dots,\,N_\text{time}\}$.
    \State Generate $j_\text{max}$ random probability distributions $\{P_j\}$, where ${P_j = \{p_i^j\,|\,i = 1,\,2,\,\dots,\,N_\text{time}\}}$.
    \State Calculate $F_{\text{res},0} = \max_i F_{\text{res}} \bigl(\mathcal{O}_i,\, \sum_k p_k^1 \bm{X}_\text{t}(t_k)\bigr)$.
    \State Set $P_\text{\tiny{PTM}} = P_1$.
    \vspace*{0.1cm}
    \For {$j = 2,\,3,\,\dots j_\text{max}$}
    \vspace*{0.1cm}
    \State Calculate $F_{\text{res},1} = \max_i F_\text{res}\bigl(\mathcal{O}_i,\,\sum_k p_k^j \bm{X}_\text{t}(t_k)\bigr)$.
    \vspace*{0.1cm}
    \If {$F_{\text{res},1} <F_{\text{res},0}$}
    \State Update $P_\text{\tiny{PTM}} = P_j$ and $F_{\text{res},0} = F_{\text{res},1}$.
    \EndIf
    \EndFor
    \end{algorithmic}
    \end{algorithm}
\end{minipage}
\par
}
One may find better heuristic methods for this optimization.

The QRPE protocol with PTM is altered as:
\begin{enumerate}
  \item Perform a one-time estimation of the training states of a two-node reservoir and load the training data $\{\bm{X}_\text{p}(t_k)\}$ to classical memory.
  \item Given properties $\{\mathcal{O}_i\}$, calculate weights $\{W_i\}$ with the training data. Run the PTM algorithm to obtain $P_\text{\tiny{PTM}}$ and $\text{Var}_0$.  Calculate $N_\text{sample}$ with the given confidence $1-\delta$ and additive error $\epsilon$. 
  \item Process $N_\text{sample}$ i.i.d. copies of the unknown state $\sigma$ with QRP. For each copy the reservoir evolution time is randomly chosen from $\{t_k\}$ with the probability distribution $P_\text{\tiny{PTM}}$. Load the readouts $\{X_i\}$ to classical memory.
  \item Calculate the estimated values $\{\tilde{\mathcal{O}}_i\}$ with the weights and readouts.

\end{enumerate}

Numerical result shows that for random reservoir dynamics, the variance could be suppressed by PTM. For each target state, we use the training data collected at 2 different time points with a random reservoir initialization, and optimize the variance upper bound with PTM. After that, we could achieve a more efficient estimation performance, as shown in Fig.~\ref{fig:TMP}.

\begin{figure}[t]
    \includegraphics[width=.7\columnwidth]{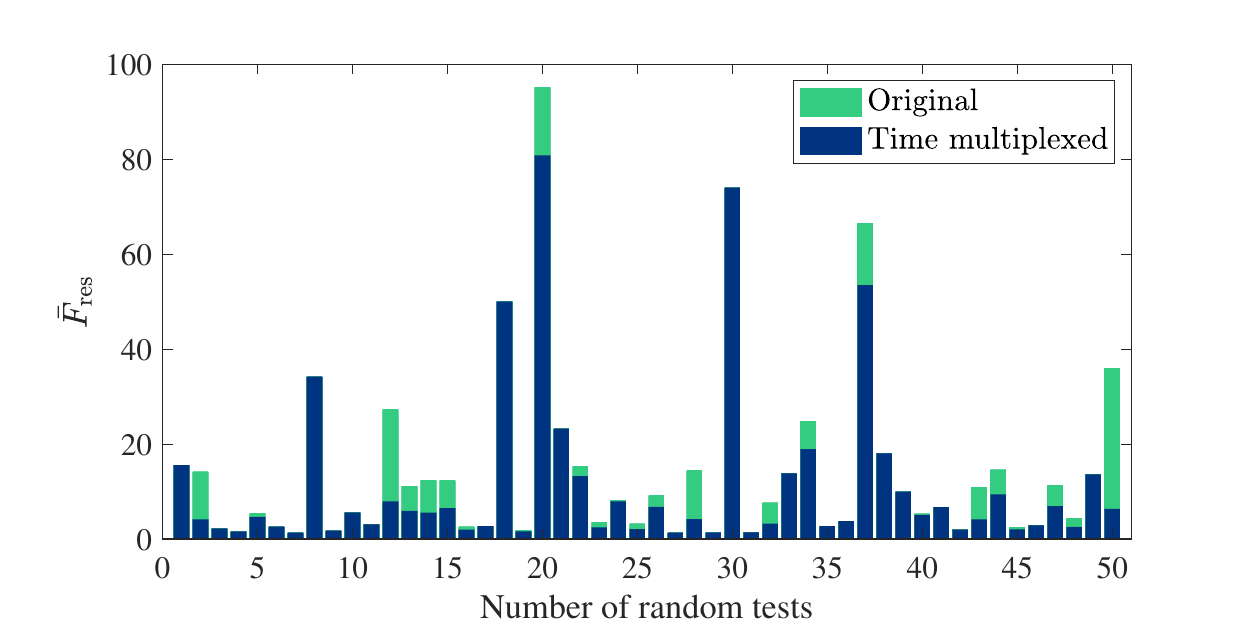}
    \caption{ Variance reduction via PTM. For each random test we use the training data collected at $t = 1$ and $t = 10$ of a two-node reservoir with parameters $J, P_i, E_i$ uniformly generated at random from $[0,5]$. The variance upper bound for overlap estimation $\bar{F}_\text{res}$ is averaged over 300 random target states. The green bars uncovered by blue bars indicate the improvement in average variance upper bound achieved by PTM.
    }
    \label{fig:TMP}
\end{figure}

\section{Estimation of nonlinear functions}
In this section we analyze the estimation of nonlinear functions with the QRPE scheme.
\subsection{Unbiased estimator for nonlinear functions}
The estimation of higher order moments, as stated in the main text, is reduced to estimating a linear function w.r.t. tensor product of the input state $\sigma$, i.e.,
\begin{equation}\label{eq:nonlinObs}
    \tr(\mathcal{O}\sigma^{\otimes m}) =  \llangle\mathcal{O}|\sigma^{\otimes m}\rrangle = \llangle\mathcal{O}|{(|\sigma\rrangle)}^{\otimes m}\,,
\end{equation}
where the second equality follows from the property of trace operation
\begin{equation}
    \tr(A\otimes B) = \tr(A) \tr(B)\,.
\end{equation} 
We estimate the  expectation values of readout operators at an evolution time $t$,
\begin{equation}
\begin{aligned}
      \bm{X} &= \bigl[\bar{X}_{\varrho_1},\,\bar{X}_{\varrho_2},\,\dots,\,\bar{X}_{\varrho_{d^2}}\bigr]
      = \bm{\mathcal{T}} \bm{M}_\text{t}\,.
\end{aligned}
\end{equation}
The combined training data for estimating Eq.~\eqref{eq:nonlinObs} is
\begin{equation}
\begin{aligned}
      \bm{X}^{\otimes m} &= \bm{\mathcal{T}}^{\otimes m} \bm{M}_\text{t}^{\otimes m}\,,
\end{aligned}
\end{equation}
and the target readout vector is
\begin{equation}
    Y^\text{tar} = \llangle\mathcal{O}| \bm{M}_\text{t}^{\otimes m}\,.
\end{equation}
Thus, the weight vector is 
\begin{equation}\label{App: eq: nonlinweight}
\begin{aligned}
W^{(m)} &= Y^\text{tar}{(\bm{X}^{\otimes m})}^{-1}
= \llangle\mathcal{O}|{(\bm{\mathcal{T}}^{\otimes m})}^{-1}\,,
\end{aligned}
\end{equation}
and the readout vector of an arbitrary input state $\sigma$ satisfies
\begin{equation}
    \bar{X}_{\sigma^{\otimes m}} = \bm{\mathcal{T}}^{\otimes m} |\sigma\rrangle^{\otimes m} = \bar{X}_{\sigma}^{\otimes m} \,.
\end{equation}
In conclusion, we obtain
\begin{equation}
    W^{(m)} \bar{X}_{\sigma}^{\otimes m} = \llangle\mathcal{O}|\sigma\rrangle^{\otimes m} = \tr(\mathcal{O}\sigma^{\otimes m})\,.
\end{equation}
For simplicity, we will neglect the superscript of $ W^{(m)} $ when there is no ambiguity.

For pair-wise interaction, the classical resources consumed in the training phase is significantly reduced. Suppose the property is decomposed into
\begin{equation}
    \mathcal{O} = \sum_i \bigotimes_{j = 1}^n \bigotimes_{k = 1}^m \mathcal{O}_{i,j,k}\,,
\end{equation}
then the weight vector is
\begin{equation}
\begin{aligned}
W &= \llangle\mathcal{O}|{(\bm{\mathcal{T}}^{\otimes m})}^{-1} =  \sum_i \bigotimes_j W_{i,j} \,,
\end{aligned}
\end{equation}
where $ W_{i,j} = \bigotimes_{k = 1}^m \llangle\mathcal{O}_{i,j,k}| \bm{\mathcal{T}}_s^{-1}$. Due to the tensor product structure of the single snapshot readout $X$, only $ W_{i,j,k}$ that correspond to nontrivial $\mathcal{O}_{i,j,k}$ are required in the classical post-processing.

\subsection{Variance of U-statistics estimators}
Suppose $N$ copies of $\sigma$ are injected to the reservoir and the corresponding readout vectors are $\{X_i| i=1,\,2,\,\dots N\}$. The uniformly minimal-variance unbiased estimator (U-statistics estimator) for estimating $W\bar{X}_{\sigma}^{\otimes m}$ is
\begin{equation}\label{eq: USest}
    \hat{\Omega}^{(m)} = \frac{1}{P(N,m)}\sum_{i_1,\,i_2,\,\dots,\, i_m}^{*} W \bigotimes_{j = 1}^m X_{i_j}\,,
\end{equation}
where $P(N,m)$ is the $m$-permutations of $N$, $\sum^*$ denotes the summation over all distinct subscripts, i.e., $\{i_1,\,i_2,\,\dots,\, i_m\}$ is an $m$-tuple of indices from the set $\{1,\,2,\,\dots\,N\}$ with distinct entries. The kernel of $\hat{\Omega}$ is a symmetric function
\begin{equation}
    h(X_{i_1},\,X_{i_2},\,\dots,\, X_{i_m}) = \frac{1}{m!}\sum_{\{l_1,\,l_2,\,\dots,\, l_m\}\in \mathcal{P}(i_1,\,i_2,\,\dots,\, i_m)} W \bigotimes_{j = 1}^m X_{l_j}\,,
\end{equation}
where $\mathcal{P}(i_1,\,i_2,\,\dots,\, i_m)$ is the set that contains all permutations of $i_1,\,i_2,\,\dots,\, i_m$. 
Let 
\begin{equation}
    h_k(x_{1},\,x_{2},\,\dots,\, x_{k}) = \mathbb{E}\bigl[ h(x_{1},\,x_{2},\,\dots,\,x_{k},\,X_{k+1},\,\dots,\, X_{m})\bigr]\,,
\end{equation}
The variance for $\hat{\Omega}^{(m)}$ in Eq.~\eqref{eq: USest} is given by the Hoeffding's theorem~\cite{Hoeffding1948}
\begin{equation}
    \text{Var}\bigl(\hat{\Omega}^{(m)}\bigr) = \tbinom{N}{m}^{-1}\sum_{k = 1}^m{\tbinom{m}{k}\tbinom{N-m}{m-k} \text{Var}\bigl[h_k(X_{1},\,X_{2},\,\dots,\, X_{k})\bigr]}\,.
\end{equation}
To estimate quadratic functions, the variance upper bound is given by Lemma S5 of Ref.~\cite{Huang2020}. Here we rephrase it as
\begin{lemma}\label{lemma:U-estVar}
The variance associated with $\hat{\Omega}^{(2)}$ satisfies \begin{equation}
    \text{Var}\bigl(\hat{\Omega}^{(2)}\bigr) \le \frac{8\mathcal{A}^{(2)}}{N} \,,
\end{equation}
where
\begin{equation}
    \mathcal{A}^{(2)} = \max{\Bigl( \text{Var} \bigl[W (X_1\otimes \bar{X})\bigr],\,\text{Var} \bigl[W (\bar{X}\otimes X_2 )\bigr],\,\sqrt{\text{Var} \bigl[W (X_1\otimes X_2)\bigr]}\Bigr)}\,.
\end{equation}
\end{lemma}
\begin{proof}
\begin{equation}
\begin{aligned}
    \text{Var}\bigl(\hat{\Omega}^{(2)}\bigr) & = \frac{2}{N(N-1)}
    \Bigl[ 2(N-2) \text{Var} \bigl[h_1(X_{1})\bigr] + \text{Var} \bigl[h_2(X_{1},\,X_{2})\bigr]\Bigr] \\
    & = \frac{1}{N(N-1)}
    \Bigl[ (N-2) \text{Var} \bigl[W(X_1\otimes \bar{X}+\bar{X}\otimes X_2)\bigr] + \frac{1}{2}\text{Var} \bigl[W(X_{1}\otimes X_{2}+X_{2}\otimes X_{1})\bigr]\Bigr] \\
\end{aligned}
\end{equation}
Note that $\text{Var}(A+B)\le 2 [\text{Var}(A)+\text{Var}(B)]$, we have
\begin{equation}
  \begin{aligned}
    \text{Var}\bigl(\hat{\Omega}^{(2)}\bigr) \le & \frac{2(N-2)}{N(N-1)}
    \Bigl[  \text{Var} \bigl[W(X_1\otimes \bar{X})\bigr]+\text{Var} \bigl[W(\bar{X}\otimes X_2)\bigr]\Bigr] \\
    & + \frac{2}{N(N-1)}\text{Var} \bigl[W(X_{1}\otimes X_{2})\bigr] \\
    \le & \frac{4}{N^2} \text{Var} \bigl[W (X_1\otimes X_2)\bigr] + \frac{2}{N} \text{Var} \bigl[W (X_1\otimes \bar{X})\bigr] + \frac{2}{N} \text{Var} \bigl[W  (\bar{X} \otimes X_1)\bigr]\,.
\end{aligned}  
\end{equation}
Then we have
\begin{equation}
    \text{Var}\bigl(\hat{\Omega}^{(2)}\bigr) \le \frac{4\mathcal{A}^{(2)}}{N} + \frac{4{\mathcal{A}^{(2)}}^2}{N^2} \,.
\end{equation}
It is reasonable to assume ${N}>{\mathcal{A}^{(2)}}$, so we have $\text{Var}\bigl(\hat{\Omega}^{(2)}\bigr) \le 8{\mathcal{A}^{(2)}}/{N} $.
\end{proof}

Next, we could compute upper bounds for $\mathcal{A}^{(2)}$ with numerical methods. Define the matrix $\bm{W}_{12}$ as\begin{equation}
    W \cdot (X_1\otimes X_2) = X_1^{T} \bm{W}_{12} X_2 = \sum_{i = 1}^{N_\text{read}}\sum_{j=1 }^{N_\text{read}} w_{ij} X_{1}(i) X_{2}(j)\,,
\end{equation}
we have
\begin{equation}
\begin{aligned}
    \text{Var} \bigl[W (X_1\otimes X_2)\bigr] &= \text{Var}\sum_{i = 1}^{N_\text{read}}\sum_{j=1}^{N_\text{read}} w_{ij} \hat{x}_i \hat{x}_j \le  W\odot W \cdot \bar{X}\otimes \bar{X}\,.
\end{aligned}
\end{equation}
Also, there are
\begin{equation}
\begin{aligned}
    \text{Var} \bigl[W (X_1\otimes \bar{X})\bigr] &= \text{Var}\sum_{i = 1}^{N_\text{read}}\sum_{j=1 }^{N_\text{read}} w_{ij} \langle \hat{o}_j \rangle \hat{x}_i \le \bar{X}^{T} (\bm{W}_{12} \bar{X})\odot (\bm{W}_{12} \bar{X})\,,
\end{aligned}
\end{equation}
and similarly
\begin{equation}
    \text{Var} \bigl[W (\bar{X}\otimes X_2)\bigr] \le  \sum_i\sum_k\sum_h w_{ki}w_{hi} \langle \hat{o}_k \rangle \langle \hat{o}_h \rangle \langle \hat{o}_i \rangle =  (\bar{X}^{T} \bm{W}_{12})\odot (\bar{X}^{T} \bm{W}_{12}) \bar{X}\,.
\end{equation}

To conclude this section, we illustrate the median of U-statistics estimators~\cite{Huang2020} by the following lemma:
\begin{lemma}
For target properties $\{\mathcal{O}_i| i = 1,\,2,\,\dots,\,M\}$ and training data ${\bm{X}_\text{t}}$, set 
\begin{equation}
    N_\text{sample} = \frac{544}{\epsilon^2} \ln\bigl(\frac{2M}{\delta}\bigr)\max_{i} \mathcal{A}^{(2)}({\mathcal{O}_i},{\bm{X}_\text{t}})\,.
\end{equation}
Then, consuming $N_\text{sample}$ i.i.d. copies of input states suffice to construct median of U-statistics estimators $\bigl\{\hat{\Omega}^{(2)}_{i}\bigr\}$, satisfying
\begin{equation}
    \bigl|\hat{\Omega}^{(2)}_{i} - \tr(\mathcal{O}_i\sigma^{\otimes 2})\bigr|\le\epsilon\quad\text{for}\quad i = 1,\,2,\dots,\,M\,,
\end{equation}
with probability no less than $1-\delta$.
\end{lemma}
\begin{proof}
     We divide the copies into $K = 2\ln(2M/\delta)$ equal-sized sets, and compute the U-statistics estimators for each set. Then, the medians of U-statistics estimators are the final estimation results for the properties. The property of median estimator ensures that if each U-statistic estimator has a variance no larger than $\epsilon^2/34$, then 
     \begin{equation}\label{eq:C26}
         \Pr\Bigl[\bigl|\hat{\Omega}^{(2)}_{i} - \tr(\mathcal{O}_i\sigma^{\otimes 2})\bigr|\le\epsilon\Bigr]\ge 1 - 2 e^{-K/2}\,.
     \end{equation}
     Thus, for each $\hat{\Omega}^{(2)}_{i}$ the confidence level is no less than $1-\delta/M$. The union bound ensures that the over all confidence level for estimating $M$ properties is no less than $1-\delta$. From Lemma.~\ref{lemma:U-estVar}, we choose the size for each set as $272/{\epsilon^2} \max_{i} \mathcal{A}^{(2)}({\mathcal{O}_i},{\bm{X}_\text{t}})$. Consequently, the total number of copies consumed is $N_\text{sample} = {544}/{\epsilon^2} \ln(2M/\delta)\max_{i} \mathcal{A}^{(2)}({\mathcal{O}_i},{\bm{X}_\text{t}})$.
\end{proof}

\end{document}